\documentclass[a4paper,11pt]{amsart}

\usepackage[colorlinks=true, linkcolor=red, linktoc=page]{hyperref}
\usepackage{amssymb}
\usepackage{enumitem}
\usepackage{url}
\usepackage{xcolor}
\usepackage{graphicx}
\usepackage{geometry}

\usepackage{algorithm}
\usepackage{algorithmic}

\usepackage[T1]{fontenc}
\usepackage[english]{babel}
\usepackage{times}
\usepackage{ucs}
\usepackage[utf8x]{inputenc}

\usepackage{dsfont}
\usepackage[mathscr]{euscript}
\usepackage{stmaryrd}

\makeatletter

\def\section{\@startsection{section}{1}%
\z@{.7\linespacing\@plus\linespacing}{.5\linespacing}%
{\normalfont\bfseries\centering}}

\def\@settitle{\begin{center}%
  \baselineskip14\p@\relax
    \bfseries
    \LARGE\@title
  \end{center}%
}

\def\@setauthors{%
  \begingroup
  \trivlist
  \centering\footnotesize \@topsep30\p@\relax
  \advance\@topsep by -\baselineskip
  \item\relax
  \andify\authors
  \def\\{\protect\linebreak}%
 {\Large\authors}%
  \endtrivlist
  \endgroup
}

\def\maketitle{\par
  \@topnum\z@ 
  \@setcopyright
  \thispagestyle{firstpage}
  \ifx\@empty\shortauthors \let\shortauthors\shorttitle
  \else \andify\shortauthors
  \fi
  \@maketitle@hook
  \begingroup
  \@maketitle
  \toks@\@xp{\shortauthors}\@temptokena\@xp{\shorttitle}%
  \toks4{\def\\{ \ignorespaces}}
  \edef\@tempa{%
    \@nx\markboth{\the\toks4
      \@nx{\the\toks@}}{\the\@temptokena}}%
  \@tempa
  \endgroup
  \c@footnote\z@
  \def\do##1{\let##1\relax}%
  \do\maketitle \do\@maketitle \do\title \do\@xtitle \do\@title
  \do\author \do\@xauthor \do\address \do\@xaddress
  \do\email \do\@xemail \do\curraddr \do\@xcurraddr
  \do\commby \do\@commby
  \do\dedicatory \do\@dedicatory \do\thanks \do\thankses
  \do\keywords \do\@keywords \do\subjclass \do\@subjclass
}
\makeatother

\newtheorem{defi}{Definition}[section]
\newtheorem{lem}[defi]{Lemma}

\newtheorem{prop}[defi]{Proposition}
\newtheorem{theo}[defi]{Theorem}
\newtheorem{cor}[defi]{Corollary}
\newtheorem{rk}[defi]{Remark}

\newcommand{\Exp}{\mathbb{E}}
\newcommand{\Var}{\mathrm{Var}}
\renewcommand{\Pr}{\mathbb{P}}
\newcommand{\R}{\mathbb{R}}
\newcommand{\N}{\mathbb{N}}

\newcommand{\dd}{\mathrm{d}}

\newcommand{\ee}{\mathrm{e}}

\renewcommand{\hat}{\widehat}
\renewcommand{\tilde}{\widetilde}
\renewcommand{\bar}{\overline}

\renewcommand{\emptyset}{\varnothing}

\newcommand{\dis}{\displaystyle}
\newcommand{\eps}{\varepsilon}

\newcommand{\Sevu}{\mathrm{s}^u}
\newcommand{\mQ}{\mathcal{Q}}
\newcommand{\mI}{\mathcal{I}}

\title[Probabilistic formulation of Miner's rule and application to structural fatigue]{Probabilistic formulation of Miner's rule\\ and application to structural fatigue}
\author[F.-B. Cartiaux, A. Ehrlacher, F. Legoll, A. Libal and J. Reygner]{Fran\c cois-Baptiste Cartiaux$^1$, Alain Ehrlacher$^2$, Frédéric Legoll$^{2,3}$,\\ Alex Libal$^{2,4}$ and Julien Reygner$^4$\\ \vskip 0cm
{\footnotesize $^1$ OSMOS Group, Puteaux, France}\\
{\footnotesize $^2$ Navier, École des Ponts, Univ Gustave Eiffel, CNRS, Marne-La-Vall\'ee, France}\\
{\footnotesize $^3$ MATHERIALS project-team, Inria, Paris, France}\\
{\footnotesize $^4$ CERMICS, École des Ponts, Marne-La-Vall\'ee, France}\\ \vskip 0cm
{\footnotesize \tt cartiaux@osmos-group.com, \{alain.ehrlacher,frederic.legoll,ales.libal,julien.reygner\}@enpc.fr}}
\date{\today}
\thanks{}

\begin{document}
   
\begin{abstract}
The standard stress-based approach to fatigue is based on the use of S-N curves. They are obtained by applying cyclic loading of constant amplitude $S$ to identical and standardised specimens until they fail. The S-N curves actually depend on a reference probability $p$: for a given cycle amplitude $S$, they provide the number of cycles at which a proportion $p$ of specimens have failed. Based on the S-N curves, Miner's rule is next used to predict the number of cycles to failure of a specimen subjected to cyclic loading with variable amplitude. In this article, we present a probabilistic formulation of Miner's rule, which is based on the introduction of the notion of health of a specimen. We show the consistency of that new formulation with the standard approaches, thereby providing a precise probabilistic interpretation of these. Explicit formulas are derived in the case of the Weibull--Basquin model. We next turn to the case of a complete mechanical structure: taking into account size effects, and using the weakest link principle, we establish formulas for the survival probability of the structure. We illustrate our results by numerical simulations on a I-steel beam, for which we compute survival probabilities and density of failure point. We also show how to efficiently approximate these quantities using the Laplace method.
\end{abstract}

\maketitle

\baselineskip=16pt


\section{Introduction}

\subsection{Motivation and main results} 

This article introduces a probabilistic framework to assess the fatigue life of a mechanical structure. Two sources of randomness are taken into account: the initial state of the structure, which may contain cracks and other flaws at various scales that cannot be observed, and future loading, which is unknown by definition. The former is of mechanical nature, and the main focus of this article is on its modelling. The latter is of statistical nature, and it shall be dealt with by standard Monte Carlo methods. 

The notions and methods introduced in this article find their roots in the standard stress-based approach to fatigue. Two main objectives are pursued:
\begin{itemize}
    \item[(i)] recast the basic notions of this approach in a probabilistic setting;
    \item[(ii)] in this setting, design methods allowing to assess the fatigue life of a structure from experimental data on fatigue testing of specimen, which may be found in standards such as Eurocodes.
\end{itemize}
These two objectives are described in more details in the next sections.

\subsubsection{S-N curve and Miner's rule}

The stress-based approach to fatigue is based on the use of S-N curves, also called Wöhler curves~\cite{Woh70}, which are available in standards such as Eurocodes\footnote{See e.g. Figure 7.1 in the standard "EN 1993-1-9:2005, Eurocode 3: Design of steel structures - Part 1-9: Fatigue".}. They are obtained by applying constant amplitude cyclic loading to identical and standardised specimens until they fail. The amplitude of each cycle is summarised by a scalar positive quantity $S$, which is expressed in $\mathrm{MPa}$, and called the \emph{severity} of the cycle. The S-N curve then represents the number of cycles to failure (NCF), denoted by $N$, as a function of $S$. In fact, it is commonly observed in such tests that, for a given value of $S$, and even with initially identical specimens, there is a significant statistical dispersion in the NCF of the sample. This is due to the presence of randomly distributed flaws in each specimen, to which we shall refer as the \emph{randomness on the initial state} of the specimen. Therefore, the S-N curve actually depends on a (sometimes implicit) reference probability $p \in (0,1)$ and represents the \emph{quantile of  order $p$}, which we shall denote by $N_p(S)$, of the distribution of the NCF in the sample. In other words, $N_p(S)$ is the number of cycles at which a proportion $p$ of specimens subjected to cyclic loading with severity $S$ have failed. An illustrative example, obtained by numerical simulation in an idealised setting, is provided in Figure~\ref{fig:SN-simu}, and similar pictures can be found in the literature, for instance in~\cite[Figures~1.1 and~1.2]{CasFer09}. The family of S-N curves for $p$ varying in $(0,1)$ is sometimes referred to as the \emph{S-N field}~\cite{Bar19}.

\begin{figure}[ht]
    \centering
    \includegraphics[width=.8\textwidth]{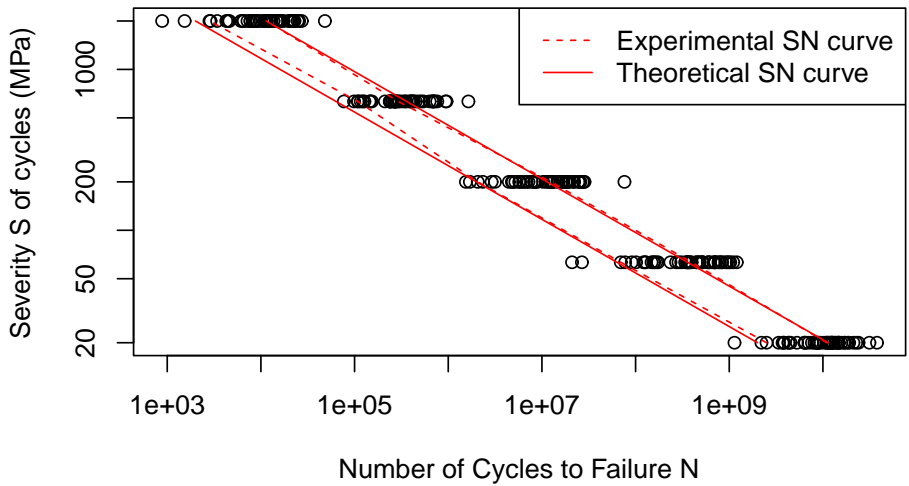}
    \caption{A numerical simulation of a fatigue testing experiment, in the idealised setting where $N$ follows a Weibull--Basquin distribution (see~\eqref{eq:Basquin:3}; we have chosen the parameters $m=1.5$, $\alpha=3$ and $\kappa$ given by~\eqref{eq:def_kappa} with $p = 0.05$, $\mathsf{N}_p = 2 \cdot 10^6$ and $\mathsf{S}_p = 200~\mathrm{MPa}$; see also Section~\ref{ss:poutre-ref}). For five levels of severity $S$, a sample of $50$ specimens are tested and their NCF are represented with points on the figure. The experimental S-N curves are plotted in dotted red lines, for the values $p=0.05$ and $p=0.5$ of the reference probability, and the theoretical S-N curve (which interpolates the theoretical quantiles of the associated Weibull distribution) is superposed in solid red line.} 
    \label{fig:SN-simu}
\end{figure}

Based on S-N curves, to predict the NCF of a specimen which is subjected to cyclic loading with variable (and possibly random) severity, a simple and widespread method is the \emph{Miner rule}, also called \emph{Palmgren--Miner} rule~\cite{Pal24,Min45}. This approach postulates the existence of a quantity $D_p$, called the \emph{cumulative damage}, which is initially equal to $0$ and which increases by $1/N_p(S)$ at each cycle of severity $S$, with $N_p(S)$ given by the S-N curve. The theoretical failure is reached when $D_p=1$, which yields a theoretical NCF $N_p(\mathbf{S})$ which depends on the whole sequence $\mathbf{S}=(S_i)_{i \geq 1}$ of severities undergone, as well as on the reference probability $p$ used to plot the S-N curve. 

\subsubsection{Probabilistic formulation of Miner's rule}

This theoretical NCF lacks a precise probabilistic interpretation, in terms of the randomness on the initial state of the specimen. The first main contribution of this work, summarised in Theorem~\ref{theo:main}, consists in providing such an interpretation: the true NCF is modelled as a random variable, whose quantile of order $p$, given $\mathbf{S}$, is shown to coincide with $N_p(\mathbf{S})$. This result, which is consistent with the case of tests with constant severity, relies on the introduction of an abstract notion of \emph{health} of a specimen, which is a random variable describing the state of the specimen after each cycle and whose existence and properties follow from first principles.

Theorem~\ref{theo:main} and its derivation are primarily of theoretical interest. However, as an operational byproduct, they yield an explicit formula for the probability that the specimen fails after $n$ cycles when subjected to cyclic loading with variable severity. Hence, a complete probabilistic formulation of Miner's rule is provided.

\subsubsection{Survival probability of structures}

The second main contribution of this work is the derivation of a formula for the survival probability of a {\em structure}, subjected to cyclic loading with variable in time and non-uniform in space severity. It is based on a continuum mechanics approach, in which the structure is partitioned into elementary volumes, which are sufficiently small so that the severity may be assumed to be uniform within each of these. These volumes are considered as independent specimens to which the probabilistic formulation of Miner's rule is applied. Using Weibull's \emph{weakest link principle}~\cite{Wei51}, we show that, in the limit of infinitesimally small elements, this approach yields size effects for the NCF of the structure, from which the analytic formula for the survival probability is derived in Theorem~\ref{theo:survival-struct}.

In this approach, the extension of the abstract notion of health from specimens to structures, again based on first principles, plays a central role. Indeed, as is the case for Miner's rule, it ensures the consistency of the description of the behaviour of the structure when subjected to cyclic loading with variable severity by extrapolating from experimental results obtained for specimens subjected to cyclic loading with constant severity.

\subsection{Outline of the article}

The probabilistic formalisation and reformulation of Miner's rule is described in Section~\ref{s:miner}. The notion of health of a specimen, on which the theoretical results of Section~\ref{s:miner} are based, is introduced in Section~\ref{s:health}. These notions are applied to the survival probability of structures in Section~\ref{s:continuum}, and a complete illustrative example is presented in Section~\ref{s:poutre}. Appendix~\ref{app:laplace} contains complements on the numerical approximation of an integral arising in the survival probability of structures.

\subsection{Related works} The random nature of the NCF in fatigue experiments has been widely observed and it is commonly accepted that a probabilistic formalism is necessary to handle this phenomenon. The probabilistic modelling of fatigue is therefore a well-investigated issue and many references are dedicated to the derivation of statistical models for S-N fields and associated estimation methods, see for instance~\cite{Cas06,Thi08,Fou14,Bar19,Mir20}. Our first result is complementary to these works, since we do not assume nor derive a particular model for the S-N field: taking the latter as an input, we actually provide a probabilistic interpretation of Miner's damage and NCF under minimal assumptions.

The second part of this article, which aims at assessing the fatigue life of a structure given experimental results on specimens, relies on statistical principles (independence and weakest link) which are commonly employed in probabilistic models for fatigue. In particular, our results share common features with the methodology developed by Castillo and Fern\'andez-Canteli, which is extensively summarised in the monograph~\cite{CasFer09}. As far as the stress-based approach to fatigue is concerned, the starting point of these works is the first-principle derivation of a general statistical model for S-N fields~\cite[Chapter~2]{CasFer09}, under a normalisation assumption which we shall also make, see Section~\ref{ss:proba}. Similarly to our argument in Section~\ref{s:continuum}, the weakest link principle then entails the derivation of size effects for extrapolation from specimens to structures~\cite[Chapter~3]{CasFer09}. Last, damage measures and damage accumulation are discussed from a probabilistic point of view in~\cite[Chapters~6 and~7]{CasFer09}. 

There are however several differences between these works and our approach. First, while~\cite{CasFer09} relies on extreme value theory to suggest to model the law of NCF as a Weibull or Gumbel distribution, we emphasize that this is not necessary and that size effects may be described in a nonparametric way, independently from extreme value distributions. A similar remark concerning size effects induced by the weakest link principle is made by Jeulin in~\cite[p.~780]{Jeu21}, based on a microscopic approach which we shall discuss in Section~\ref{sss:Poisson}. Second, our study of damage accumulation remains in the strict framework of Miner's rule and does not attempt to provide other damage measures. Last, and foremost, the introduction in the present article of the notion of health of a specimen or of a structure, which is the key theoretical feature of our work, is new. It brings forth analytical formulas for the survival probability of a structure under variable, nonuniform and random loading, and therefore entails to complement the usual computation of the deterministic fatigue life of the structure with the whole distribution of the random fatigue life.

\subsection{Notation and comments}\label{ss:discr-cont} 

Throughout the article, loading cycles are indexed by the set of integers $\N = \{1, 2, \ldots\}$, so that the NCF of a specimen (or a structure) is, in principle, an integer-valued random variable. However, in order to not overload the article with technical aspects, we shall often implicitly consider such variables as continuous, which is legitimate since, in practice, they assume large values. A consequence of this notational abuse is that we shall (still implicitly) consider that the distribution functions of discrete quantities are continuous and increasing, and that discrete dynamical systems are well approximated by continuous ones. We denote by the bold symbol $\mathbf{S} = (S_i)_{i \geq 1}$ a sequence of cycle severities (where $S_i$ is of course the severity of the $i$-th cycle). 

\section{Probabilistic formalisation of Miner's rule}\label{s:miner}

Section~\ref{ss:proba} introduces a standard probabilistic formalism to model fatigue testing, as well as the set of Assumptions~\eqref{ass:A1}, \eqref{ass:A2} and~\eqref{ass:A3}, which will serve as first principles. The probabilistic interpretation of Miner's rule, derived from these first principles, is presented in Section~\ref{ss:miner-proba}. A particular parametric case, the Weibull--Basquin model, is introduced in Section~\ref{ss:weibull-basquin}.

\subsection{Probabilistic setting and modelling assumptions}\label{ss:proba} 

\subsubsection{Probabilistic setting} 

From the statistical point of view, the experiment of cyclic loading of a specimen with variable severity results in the observation of a random variable
\begin{equation*}
  N : \left\{\begin{array}{rcl}
    \Omega \times \mathcal{X} & \to & \N\\
    (\omega, \mathbf{S}) & \mapsto & N(\omega, \mathbf{S})
  \end{array}\right.
\end{equation*}
which describes the number of cycles after which the specimen fails when subjected to the (deterministic) sequence of cycle severities $\mathbf{S} = (S_i)_{i \geq 1} \in (0,+\infty)^\N =: \mathcal{X}$. The extra parameter $\omega \in \Omega$ accounts for randomness associated with each realisation of the experiment; as is standard in probabilistic modelling, the space $\Omega$ is not given a precise physical signification but must rather be interpreted as an abstract space of possible outcomes, endowed with a probability measure $\Pr$ assessing the likelihood of each outcome. For any $\mathbf{S} \in \mathcal{X}$, the \emph{law} $\mathrm{P}_\mathbf{S}$ of the variable $N(\omega, \mathbf{S})$ is then the image of $\Pr$ by the mapping $\omega \mapsto N(\omega, \mathbf{S})$; it is thus a probability measure on $\N$ which depends on $\mathbf{S}$.

With this model at hand, fatigue testing with cyclic loading of constant severity $S$ corresponds to observing realisations of the random variable $N(\omega,\mathbf{S})$ when $S_i=S$ for any $i \geq 1$. We use the shorthand notation $N(\omega,S):=N(\omega,\mathbf{S})$ in this case and denote by $\mathrm{P}_S$ the law of this variable. Therefore, experimental data yield the family of \emph{marginal distributions} $\{\mathrm{P}_S, \, S>0\}$. There are several idealisations in this statement, since in practice finitely many specimens are tested with finitely many values of $S$, hence describing the whole set of laws $\{\mathrm{P}_S, \, S>0\}$ from experiments requires to handle both statistical uncertainty related with the estimation of $\mathrm{P}_S$ from a finite sample and extrapolation to untested values of $S$. Still, from a purely theoretical point of view, even a perfect knowledge of $\{\mathrm{P}_S, \, S>0\}$ does not allow to characterise the law of $N(\omega,\mathbf{S})$ for nonconstant sequences $\mathbf{S} \in \mathcal{X}$. To carry out this task, it is thus necessary to introduce modelling assumptions on the collection of random variables $\{N(\omega,\mathbf{S}), \, \mathbf{S} \in \mathcal{X}\}$.

\subsubsection{Modelling assumptions} 

At the informal level, our two main modelling assumptions on the collection of random variables $\{N(\omega,\mathbf{S}), \, \mathbf{S} \in \mathcal{X}\}$ write as follows:
\begin{enumerate}[label=(A\arabic*),ref=A\arabic*]
  \item\label{ass:A1} Randomness only originates in the initial state of the specimen.
  \item\label{ass:A2} The damage caused by a cycle is a deterministic function of the current state of the specimen and of the severity of the cycle.
\end{enumerate}
These two assumptions are complemented by a shape homogeneity assumption on the probability measures $\{\mathrm{P}_S, \, S>0\}$:
\begin{enumerate}[label=(A\arabic*),ref=A\arabic*,start=3]
  \item\label{ass:A3} The law of the random variable $N(\omega,S)$ only depends on $S$ through a multiplicative factor: there exists a \emph{scale} function $\langle N \rangle : (0,+\infty) \to (0,+\infty)$ and a \emph{shape} function $u : [0,+\infty) \to [0,1]$ which decreases from $1$ to $0$ such that, for any $h \geq 0$ and $S>0$,
  \begin{equation*}
      \Pr\left(\frac{N(\omega,S)}{\langle N \rangle(S)} > h\right) = u(h).
  \end{equation*}
\end{enumerate}
Assumption~\eqref{ass:A3} is of a less fundamental nature than Assumptions~\eqref{ass:A1} and~\eqref{ass:A2}, but it is often observed in practice. For instance, it is present in the derivation of the probabilistic model for S-N fields performed in~\cite[Chapter~2]{CasFer09} and related to the notion of \emph{normalisation} introduced in Chapters~6 and~7 of the same reference. Furthermore, in the latter reference, extreme value theory is employed to justify the fact that Weibull or Gumbel distributions may provide adequate models for the shape function $u$. We detail the case of the Weibull model in Section~\ref{ss:weibull-basquin}. We emphasize that the main results of this work are stated in the nonparametric context of Assumption~\eqref{ass:A3}, with illustration of some of them in the Weibull case.

\begin{rk}\label{rk:norma}
Note that, if the random variable $N(\omega,S)$ satisfies Assumption~\eqref{ass:A3} for some scale function $\langle N \rangle$ and some shape function $u$, then it also does for the scale function $S \mapsto c \, \langle N \rangle(S)$ and for the shape function $h \mapsto u(c \, h)$, for any $c>0$. We will further discuss this normalisation freedom in Section~\ref{sec:norma}. 
\end{rk}

\begin{rk}[On Assumption~\eqref{ass:A3} and the S-N curves]\label{rk:SN-A3}
Let $N_p(S)$ be the quantile of order $p$ of the random variable $N(\omega,S)$, which, we recall, satisfies $\Pr\left( N(\omega,S) \leq N_p(S) \right) = p$. Assumption~\eqref{ass:A3} implies that, for a given severity $S>0$, the quantiles $N_p(S)$ of $N(\omega,S)$ are multiples of each other, namely
    \begin{equation*}
        \forall p \in (0,1), \qquad N_p(S) = u^{-1}(1-p) \ \langle N \rangle(S).
    \end{equation*}
    As a consequence, in log-log coordinates, the S-N curves associated with different reference probabilities are horizontal translations of each other, since, for any $p,q \in (0,1)$, the quantity $\ln N_p(S)-\ln N_q(S)$ does not depend on $S$.
\end{rk}

\subsection{Miner's cumulative damage}\label{ss:miner-proba}

\subsubsection{Miner's theoretical NCF}

For a given sequence of severities $\mathbf{S} \in \mathcal{X}$, we use the notation $\mathbf{S}_n=(S_1, \ldots, S_n) \in (0,+\infty)^n$. For a given reference probability $p \in (0,1)$, Miner's cumulative damage is the sequence $(D_{p,n}(\mathbf{S}_n))_{n \geq 0}$ defined by
\begin{equation}\label{eq:Dpn}
    \forall n \geq 0, \qquad D_{p,n}(\mathbf{S}_n) = \sum_{i=1}^n \frac{1}{N_p(S_i)},
\end{equation}
where $S_i$ is the severity of the $i$-th cycle and $N_p(S_i)$ is given by the S-N curve (we take the convention $\dis \sum_{i=1}^0 \cdot = 0$). Miner's theoretical NCF is then defined as
\begin{equation}\label{eq:NpbS}
    N_p(\mathbf{S}) = \inf\{n \geq 1: \, D_{p,n}(\mathbf{S}_n) \geq 1\}.
\end{equation}
When all cycles have the same severity $S$, then it is clear that $N_p(\mathbf{S})=N_p(S)$ so that Miner's theoretical NCF is the quantile of order $p$ of the random variable $N(\omega,S)$. Under Assumptions~\eqref{ass:A1}, \eqref{ass:A2} and~\eqref{ass:A3}, this interpretation is extended to the case of variable severities in Section~\ref{sec:theo_random}.

\subsubsection{Random cumulative damage and main theoretical result} \label{sec:theo_random}

Proceeding by analogy with~\eqref{eq:Dpn}, we define the \emph{random cumulative damage} as the random sequence $(D_n(\omega,\mathbf{S}_n))_{n \geq 0}$ given by 
\begin{equation}\label{eq:Dn}
  D_n(\omega,\mathbf{S}_n) = \sum_{i=1}^n \frac{1}{N(\omega,S_i)}.
\end{equation}

Our main results on the probabilistic formulation of Miner's rule are summarised as follows.

\begin{theo}[Probabilistic interpretation of Miner's cumulative damage and theoretical NCF]\label{theo:main}
  Under Assumptions~\eqref{ass:A1}, \eqref{ass:A2} and~\eqref{ass:A3}, the random variable $N(\omega,\mathbf{S})$ and the random sequence $(D_n(\omega,\mathbf{S}_n))_{n \geq 0}$ satisfy the identity
  \begin{equation}\label{eq:ND}
    \forall (\omega, \mathbf{S}) \in \Omega \times \mathcal{X}, \qquad N(\omega, \mathbf{S}) = \inf\{n \geq 1: \, D_n(\omega,\mathbf{S}_n) \geq 1\}.
  \end{equation}
  Besides, for any $p \in (0,1)$ and for any $\mathbf{S} \in \mathcal{X}$,
  \begin{enumerate}[label=(\roman*),ref=\roman*]
    \item\label{it:main:D} the quantity $D_{p,n}(\mathbf{S}_n)$ defined in~\eqref{eq:Dpn} is the quantile of order $1-p$ of $D_n(\omega,\mathbf{S}_n)$;
    \item\label{it:main:N} the quantity $N_p(\mathbf{S})$ defined in~\eqref{eq:NpbS} is the quantile of order $p$ of $N(\omega, \mathbf{S})$.
  \end{enumerate}
\end{theo}

The interest of Theorem~\ref{theo:main} is primarily theoretical, since in practice, the random cumulative damage $D_n(\omega,\mathbf{S}_n)$ cannot be computed: indeed, for a given realisation $\omega$ of the initial state of the specimen, the family of random variables $N(\omega,S_i)$, $i \geq 1$, cannot be observed. However this statement allows to provide a statistical interpretation, in terms of the random NCF $N(\omega,\mathbf{S})$ which is eventually observed, of the computable quantities $D_{p,n}(\mathbf{S}_n)$ and $N_p(\mathbf{S})$.

The formalisation of Assumptions~\eqref{ass:A1} and~\eqref{ass:A2}, together with the proof of Theorem~\ref{theo:main}, are detailed in Section~\ref{s:health}. They rely on the introduction of an abstract notion of \emph{health} of a specimen, which describes the state of the specimen and deteriorates at each cycle. Since Theorem~\ref{theo:main}~\eqref{it:main:N} shows that all quantiles of the variable $N(\omega, \mathbf{S})$ can be computed from Miner's rule, the following corollary is immediate.

\begin{cor}[Characterisation of the law of $N(\omega, \mathbf{S})$ by $\{\mathrm{P}_S, \, S>0\}$]
  Under the assumptions of Theorem~\ref{theo:main}, for any $\mathbf{S} \in \mathcal{X}$, the law of the random variable $N(\omega, \mathbf{S})$ is entirely defined by the family of laws $\{\mathrm{P}_S, \, S>0\}$ obtained by fatigue testing with constant severity.
\end{cor}

The following result is the practical consequence of Theorem~\ref{theo:main}. 

\begin{cor}[Survival probability]\label{cor:survival}
  Under the assumptions of Theorem~\ref{theo:main}, for any value of the reference probability $p$, we have
  \begin{equation}\label{eq:survival}
    \forall \mathbf{S} \in \mathcal{X}, \quad \forall n \geq 0, \qquad \Pr\left(N(\omega,\mathbf{S}) > n\right) = u\left(u^{-1}(1-p) \, D_{p,n}(\mathbf{S}_n)\right),
  \end{equation}
  where the shape function $u$ has been introduced in Assumption~\eqref{ass:A3}.
\end{cor}

As will be clarified in Section~\ref{s:health}, the function $u$ is related to statistics of the initial state of the specimen, which by Assumption~\eqref{ass:A1} is assumed to encode all randomness of the model. The formula~\eqref{eq:survival} therefore allows to relate the survival probability with Miner's cumulative damage. The proofs of Theorem~\ref{theo:main} and Corollary~\ref{cor:survival} are postponed until Section~\ref{ss:pf-main}.

\subsection{The Weibull--Basquin model}\label{ss:weibull-basquin} 

In the literature, it is often suggested to model the laws $\mathrm{P}_S$ of the variables $N(\omega,S)$ observed in fatigue testing as Weibull distributions with a common shape parameter $m>0$ (called the \emph{Weibull modulus}) and a scale parameter $\langle N \rangle(S)$ which may depend on $S$, namely
\begin{equation} \label{eq:weibull}
    \forall n \geq 0, \qquad \Pr(N(\omega,S)>n) = \exp\left(-\left(\frac{n}{\langle N \rangle(S)}\right)^m\right).
\end{equation}
It is clear that such a model satisfies Assumption~\eqref{ass:A3}, with $u(h)=\exp(-h^m)$. In this case, the S-N curve for the reference probability $p$ has equation
\begin{equation} \label{eq:lien_Np_angleN}
    N_p(S) = \left(-\ln(1-p)\right)^{1/m} \ \langle N \rangle(S),
\end{equation}
see Remark~\ref{rk:SN-A3}, and the formula~\eqref{eq:survival} rewrites, for any reference probability $p$,
\begin{equation}\label{eq:survival-weibull}
  \forall \mathbf{S} \in \mathcal{X}, \quad \forall n \geq 0, \qquad \Pr\left(N(\omega,\mathbf{S}) > n\right) = (1-p)^{(D_{p,n}(\mathbf{S}_n))^m}.
\end{equation}
This identity shows how to deduce the survival probability of a specimen subjected to variable cyclic loading from the deterministic computation of Miner's cumulative damage associated with the reference probability $p$, and assuming in addition the knowledge of the Weibull modulus $m$.

Furthermore, parametric models may also be employed to model S-N curves (see for instance a short list in~\cite[Tables~2.2 and~2.3]{CasFer09}). The simplest one is \emph{Basquin's model}~\cite{Bas10}, which postulates the existence of some $\alpha>0$ such that, in log-log plot, the S-N curve is linear with slope $-1/\alpha$, that is to say
\begin{equation}\label{eq:Basquin}
    \forall S > 0, \qquad \ln N_p(S) = -\alpha \ln S + \mathrm{Cte}.
\end{equation}
In this case, the S-N curve is entirely described by $\alpha$ and the location of a single point $(\mathsf{N}_p,\mathsf{S}_p)$ of this curve. The value $\mathsf{S}_p$ is called the \emph{detail category} of the specimen, associated with the NCF $\mathsf{N}_p$. Then~\eqref{eq:Basquin} rewrites
\begin{equation}\label{eq:Basquin:2}
    \forall S > 0, \qquad \frac{N_p(S)}{\mathsf{N}_p} = \left(\frac{S}{\mathsf{S}_p}\right)^{-\alpha}.
\end{equation}
Thus, assuming both a Weibull model for $N(\omega,S)$ with modulus $m$, and a Basquin model for the S-N curve with slope $-\alpha$, reference probability $p$ and detail category $(\mathsf{N}_p,\mathsf{S}_p)$, one gets from~\eqref{eq:lien_Np_angleN} and~\eqref{eq:Basquin:2} the identity
\begin{equation} \label{eq:def_kappa}
    \langle N \rangle(S) = \kappa(m,\alpha,p) \, S^{-\alpha}, \qquad \kappa(m,\alpha,p) = (-\ln(1-p))^{-1/m} \, \mathsf{N}_p \, \mathsf{S}_p^\alpha,
\end{equation}
for the shape parameter. Formula~\eqref{eq:survival-weibull} can be recast in a simple expression. We infer from~\eqref{eq:def_kappa} that $\dis \ln(1-p) = - \frac{\mathsf{N}_p^m \, \mathsf{S}_p^{\alpha m}}{\kappa(m,\alpha,p)^m}$ and we compute from~\eqref{eq:Dpn} and~\eqref{eq:Basquin:2} that $\dis D_{p,n}(\mathbf{S}_n) = \sum_{i=1}^n \frac{1}{N_p(S_i)} = \frac{1}{\mathsf{N}_p \, \mathsf{S}_p^\alpha} \sum_{i=1}^n S_i^\alpha$. We thus deduce from~\eqref{eq:survival-weibull} the formula
\begin{equation}\label{eq:Basquin:3}
    \Pr\left(N(\omega,\mathbf{S}) > n\right) = \exp\left(-\frac{1}{\kappa(m,\alpha,p)^m}\left(\sum_{i=1}^n S_i^\alpha\right)^m\right)
\end{equation}
for the survival probability associated with the sequence of severities $\mathbf{S}$. This formula as well as the results of Theorem~\ref{theo:main} are illustrated on Figure~\ref{fig:WB-per}, for the same setting as in the introductory example of Figure~\ref{fig:SN-simu}.

\begin{figure}[ht]
    \centering
    \includegraphics[width=.8\textwidth]{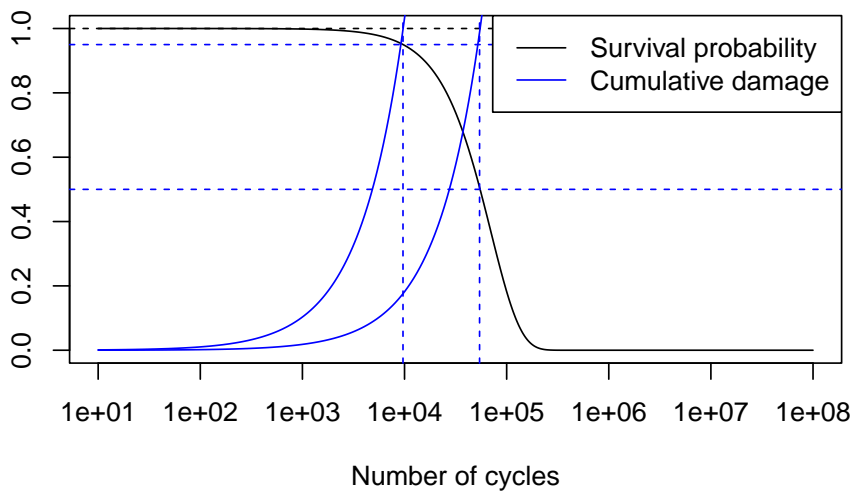}
    \caption{Simulation of a fatigue test with a non-constant sequence of severities $\mathbf{S}$, in the Weibull--Basquin model with the same parameters as on Figure~\ref{fig:SN-simu}. The solid black curve describes the survival function $\Pr(N(\omega,\mathbf{S})>n)$ as a function of the number of cycles $n$, while the two solid blue curves represent the cumulative damage $D_{p,n}(\mathbf{S}_n)$, for the two values $p=0.05$ (left-most curve) and $p=0.5$ (right-most curve) of the reference probability for which the S-N curves are plotted on Figure~\ref{fig:SN-simu}. The two dotted blue horizontal lines correspond to the probabilities 0.95 and 0.5. In accordance with the statement of Theorem~\ref{theo:main}, the NCF $N_p(\mathbf{S})$ at which the cumulative damage reaches $1$ (shown by the two dotted blue vertical lines) is observed to be the quantile of order $1-p$ of $N(\omega,\mathbf{S})$.}
    \label{fig:WB-per}
\end{figure}

\section{The notion of health of a specimen}\label{s:health}

The purpose of this section is to prove Theorem~\ref{theo:main} and Corollary~\ref{cor:survival}. To do so, Assumptions~\eqref{ass:A1} and \eqref{ass:A2} are first reformulated in more mathematical terms in Section~\ref{ss:reform}, which relies on the introduction of the notion of \emph{health} of a specimen. Based on this notion, Theorem~\ref{theo:main} and Corollary~\ref{cor:survival} are proved in Section~\ref{ss:pf-main}. 

\subsection{Formalisation of Assumptions~\eqref{ass:A1} and~\eqref{ass:A2}} \label{ss:reform} 

Our basic postulate is the existence of a random variable $H_n(\omega,\mathbf{S}_n)$, which we call the \emph{health} of the specimen after $n$ loading cycles and which describes its state. This quantity remains positive as long as the specimen does not fail, so that the sequence $(H_n(\omega,\mathbf{S}_n))_{n \geq 0}$ is related with the random NCF $N(\omega,\mathbf{S})$ of the specimen, introduced in Section~\ref{ss:proba}, through the identity
\begin{equation} \label{eq:def_N}
    \forall (\omega, \mathbf{S}) \in \Omega \times \mathcal{X}, \qquad N(\omega, \mathbf{S}) = \inf\{n \geq 1: \, H_n(\omega,\mathbf{S}_n) \leq 0\}.
\end{equation}
Then Assumptions~\eqref{ass:A1} and~\eqref{ass:A2} are reformulated as the existence of:
\begin{itemize}
    \item a positive random variable $H_0(\omega)$, which we call the \emph{initial health} of the specimen, that encodes its initial state;
    \item a deterministic, measurable function $f : \R \times (0,+\infty) \to (0,+\infty)$ which describes the decrease of health at each cycle, so that the health sequence is defined by
    \begin{equation*}
        H_0(\omega,\mathbf{S}_0) = H_0(\omega), \qquad H_{n+1}(\omega,\mathbf{S}_{n+1}) - H_n(\omega,\mathbf{S}_n) = -f(H_n(\omega,\mathbf{S}_n),S_{n+1}).
    \end{equation*}
\end{itemize}

The health sequence is not uniquely defined; indeed, it is easily checked that, for any increasing function  $\varphi : \R \to \R$ such that $\varphi(0)=0$, the sequence $(\varphi(H_n(\omega,\mathbf{S}_n)))_{n \geq 0}$ has the same properties as $(H_n(\omega,\mathbf{S}_n))_{n \geq 0}$, with a modified function $f^\varphi$. Thus, health is only defined up to an increasing transform.

The main theoretical result of this section is that the shape homogeneity Assumption~\eqref{ass:A3} entails the existence of such a transform which decreases \emph{additively}, namely a sequence $(\bar{H}_n(\omega,\mathbf{S}_n))_{n \geq 0}$ such that $\bar{H}_{n+1}(\omega,\mathbf{S}_{n+1}) - \bar{H}_n(\omega,\mathbf{S}_n)$ does not depend on $\bar{H}_n(\omega,\mathbf{S}_n)$, but only on $S_{n+1}$.

\begin{prop}[Additively decreasing health]\label{prop:barH}
    Assume that there exists a health sequence $(H_n(\omega,\mathbf{S}_n))_{n \geq 0}$ associated with some random variable $H_0(\omega)$ and function $f(h,s)$ as introduced above. Let Assumption~\eqref{ass:A3} hold, and let $\langle N \rangle(s)$ and $u(h)$ be the associated scale and shape functions. For any $h \geq 0$, set 
    \begin{equation*}
        \varphi(h) = u^{-1}(\Pr(H_0(\omega)>h)),
    \end{equation*}
    and for any $n \geq 0$ and any $(\omega,\mathbf{S}) \in \Omega \times \mathcal{X}$, define
    \begin{equation*}
        \bar{H}_n(\omega,\mathbf{S}_n) = \varphi(H_n(\omega,\mathbf{S}_n)).
    \end{equation*}
    We then have the following:
    \begin{enumerate}[label=(\roman*),ref=\roman*]
        \item\label{it:barH:1} For any $(\omega,S) \in \Omega \times (0,+\infty)$, 
        \begin{equation*}
            N(\omega,S) = \langle N \rangle(S) \ \bar{H}_0(\omega).
        \end{equation*}
        \item\label{it:barH:2} For any $(\omega,\mathbf{S}) \in \Omega \times \mathcal{X}$, the sequence $(\bar{H}_n(\omega,\mathbf{S}_n))_{n \geq 0}$ satisfies the identities
        \begin{equation} \label{eq:evol_barH}
            \bar{H}_0(\omega,\mathbf{S}_0) = \varphi(H_0(\omega)), \qquad \bar{H}_{n+1}(\omega,\mathbf{S}_{n+1})-\bar{H}_n(\omega,\mathbf{S}_n) = - \frac{1}{\langle N \rangle(S_{n+1})},
        \end{equation}
        and
        \begin{equation*}
            N(\omega,\mathbf{S}) = \inf\{n \geq 1: \, \bar{H}_n(\omega,\mathbf{S}_n) \leq 0\}.
        \end{equation*}
        \item\label{it:barH:3} The health $\bar{H}$ is related with the random cumulative damage defined in~\eqref{eq:Dn} by the identity
        \begin{equation*}
            \bar{H}_n(\omega,\mathbf{S}_n) = (1-D_n(\omega,\mathbf{S}_n)) \ \bar{H}_0(\omega).
        \end{equation*}
    \end{enumerate}
\end{prop}

\begin{rk}\label{rk:norma2}
In Remark~\ref{rk:norma}, we have pointed out some normalisation freedom in the definition of the functions $\langle N \rangle$ and $u$. This yields the fact that, if $(\bar{H}_n(\omega,\mathbf{S}_n))_{n \geq 0}$ is some health which decreases in an additive manner, then so does the health $(\tilde{H}_n(\omega,\mathbf{S}_n))_{n \geq 0}$ defined by $\tilde{H}_n(\omega,\mathbf{S}_n) = \bar{H}_n(\omega,\mathbf{S}_n)/c$ for some $c>0$.
\end{rk}

\begin{proof}
    In accordance with the discussion of Section~\ref{ss:discr-cont}, in the proof we shall use as an equality the first-order approximation
    \begin{equation}\label{eq:approx-phi}
        \phi(H_{n+1}(\omega,\mathbf{S}_{n+1})) - \phi(H_n(\omega,\mathbf{S}_n)) \simeq - \phi'(H_n(\omega,\mathbf{S}_n)) \ f(H_n(\omega,\mathbf{S}_n), S_{n+1}),
    \end{equation}
    for any smooth function $\phi$.
    
    For any $h \geq 0$ and $S>0$, let us define the quantity $\tilde{\varphi}(h,S)$ by
    \begin{equation*}
        \tilde{\varphi}(h,S) = \frac{1}{\langle N \rangle(S)}\int_0^h \frac{\dd h'}{f(h',S)}.
    \end{equation*}
    In view of~\eqref{eq:approx-phi} and of the evolution law on $H$, we have, for any $S>0$,
    \begin{equation*}
        \tilde{\varphi}(H_{n+1}(\omega,S),S) - \tilde{\varphi}(H_n(\omega,S),S) = -\frac{1}{\langle N \rangle(S)},
    \end{equation*}
    so that
    \begin{equation} \label{eq:evol_tildeH}
        \tilde{\varphi}(H_n(\omega,S),S) = \tilde{\varphi}(H_0(\omega),S) - \frac{n}{\langle N\rangle(S)}.
    \end{equation}
    We infer from the definition~\eqref{eq:def_N} of $N(\omega,S)$ and from the fact that $\tilde{\varphi}$ is increasing with respect to its first variable that
    $$
    N(\omega,S) = \inf\{n \geq 1: \, \tilde{\varphi}(H_n(\omega,S),S) \leq \tilde{\varphi}(0,S) \} = \inf\{n \geq 1: \, \tilde{\varphi}(H_n(\omega,S),S) \leq 0 \}.
    $$
    In view of~\eqref{eq:evol_tildeH}, we therefore obtain
    \begin{equation}\label{eq:pf-barH:1}
        N(\omega,S) = \langle N \rangle(S) \ \tilde{\varphi}(H_0(\omega),S).
    \end{equation}
    Since $\tilde{\varphi}(\cdot,S)$ is increasing, we deduce that, for any $h \geq 0$,
    \begin{align*}
        \Pr\left(H_0(\omega) > h\right) &= \Pr\left(\tilde{\varphi}(H_0(\omega),S) > \tilde{\varphi}(h,S)\right)\\
        &= \Pr\left(\frac{N(\omega,S)}{\langle N \rangle (S)} > \tilde{\varphi}(h,S)\right)\\
        &= u\left(\tilde{\varphi}(h,S)\right),
    \end{align*}
    where the last equality follows from Assumption~\eqref{ass:A3}. As a consequence,
    \begin{equation*}
        \tilde{\varphi}(h,S) = u^{-1}\left(\Pr\left(H_0(\omega) > h\right)\right) = \varphi(h).
    \end{equation*}
    Combining this result with the identity~\eqref{eq:pf-barH:1} proves Assertion~\eqref{it:barH:1}.
    
    To prove Assertion~\eqref{it:barH:2}, we use~\eqref{eq:approx-phi} again and the fact that $\tilde{\varphi}(h,S) = \varphi(h)$ does not depend on $S$ to write, for any $\mathbf{S} \in \mathcal{X}$,
    \begin{align*}
        \varphi(H_{n+1}(\omega,\mathbf{S}_{n+1})) - \varphi(H_n(\omega,\mathbf{S}_n)) &= -\varphi'(H_n(\omega,\mathbf{S}_n)) \ f(H_n(\omega,\mathbf{S}_n),S_{n+1})\\
        &=-\frac{\partial\tilde{\varphi}}{\partial h}(H_n(\omega,\mathbf{S}_n),S_{n+1}) \ f(H_n(\omega,\mathbf{S}_n),S_{n+1})\\
        &= -\frac{1}{\langle N \rangle(S_{n+1})}.
    \end{align*}
    Thus, Equation~\eqref{eq:evol_barH} holds with $\bar{H}_n(\omega,\mathbf{S}_n) = \varphi(H_n(\omega,\mathbf{S}_n))$. The last statement of Assertion~\eqref{it:barH:2} is a simple consequence of~\eqref{eq:def_N} and of the fact that $\varphi$ is increasing.     
    
    Assertion~\eqref{it:barH:3} then immediately follows from Assertions~\eqref{it:barH:2} and~\eqref{it:barH:1} and from the definition~\eqref{eq:Dn} of $D_n(\omega,\mathbf{S}_n)$:
    $$
    \bar{H}_n(\omega,\mathbf{S}_n) = \bar{H}_0(\omega) - \sum_{i=1}^n \frac{1}{\langle N \rangle(S_i)} = \bar{H}_0(\omega) - \sum_{i=1}^n \frac{\bar{H}_0(\omega)}{N(\omega,S_i)} = \bar{H}_0(\omega) \left( 1-D_n(\omega,\mathbf{S}_n) \right). 
    $$
    This concludes the proof of Proposition~\ref{prop:barH}.
\end{proof}

From now on, we shall refer to the sequence $(\bar{H}_n(\omega,\mathbf{S}_n))_{n \geq 0}$ as the health sequence of the specimen, and no longer use the initial sequence $(H_n(\omega,\mathbf{S}_n))_{n \geq 0}$ provided by Assumptions~\eqref{ass:A1} and~\eqref{ass:A2}.

Proposition~\ref{prop:barH}~\eqref{it:barH:3} emphasizes the link between health and damage. We point out the fact that the damage $D_n(\omega,\mathbf{S}_n)$ is initially deterministic and increases randomly, while the health $\bar{H}_n(\omega,\mathbf{S}_n)$ is initially random and decreases deterministically.

\begin{rk}\label{rk:barH0-unif}
    We deduce from Proposition~\ref{prop:barH} that, for any $h \geq 0$, we have $\Pr(\bar{H}_0(\omega)>h)=u(h)$.
\end{rk}

\subsection{Proof of Theorem~\ref{theo:main} and Corollary~\ref{cor:survival}} \label{ss:pf-main} 

Let the assumptions of Proposition~\ref{prop:barH} be in force. Then the identity~\eqref{eq:ND} in Theorem~\ref{theo:main} is an immediate consequence of Proposition~\ref{prop:barH}~(\ref{it:barH:2}-\ref{it:barH:3}). We now turn to the proof of Assertions~\eqref{it:main:D} and~\eqref{it:main:N} of Theorem~\ref{theo:main}.

For any reference probability $p \in (0,1)$, we first deduce from Remark~\ref{rk:SN-A3} and Proposition~\ref{prop:barH}~\eqref{it:barH:1} the respective identities
\begin{align*}
    &D_{p,n}(\mathbf{S}_n) = \sum_{i=1}^n \frac{1}{N_p(S_i)} = \frac{1}{u^{-1}(1-p)}\sum_{i=1}^n \frac{1}{\langle N \rangle (S_i)},\\
    &D_n(\omega,\mathbf{S}_n) = \sum_{i=1}^n \frac{1}{N(\omega,S_i)} = \frac{1}{\bar{H}_0(\omega)}\sum_{i=1}^n \frac{1}{\langle N \rangle (S_i)},
\end{align*}
from which we deduce that
\begin{equation} \label{eq:indy}
    \bar{H}_0(\omega) D_n(\omega,\mathbf{S}_n) = u^{-1}(1-p) \ D_{p,n}(\mathbf{S}_n).
\end{equation}
As a consequence, and thanks to Remark~\ref{rk:barH0-unif}, we get
\begin{equation*}
    \Pr\left(D_n(\omega,\mathbf{S}_n) \leq D_{p,n}(\mathbf{S}_n)\right) = \Pr\left(\bar{H}_0(\omega) \geq u^{-1}(1-p)\right) = 1-p,
\end{equation*}
which proves Theorem~\ref{theo:main}~\eqref{it:main:D}. We next infer from~\eqref{eq:ND} and~\eqref{eq:indy} that
\begin{equation*}
    N(\omega,\mathbf{S}) = \inf\left\{n \geq 1: \, D_{p,n}(\mathbf{S}_n) \geq \frac{\bar{H}_0(\omega)}{u^{-1}(1-p)}\right\}.
\end{equation*}
In view of~\eqref{eq:NpbS}, the fact that $N(\omega,\mathbf{S}) \leq N_p(\mathbf{S})$ is equivalent to $\dis \frac{\bar{H}_0(\omega)}{u^{-1}(1-p)} \leq 1$, and thus
\begin{equation*}
    \Pr\left(N(\omega,\mathbf{S}) \leq N_p(\mathbf{S})\right) = \Pr\left(\bar{H}_0(\omega) \leq u^{-1}(1-p)\right) = p,
\end{equation*}
which proves Theorem~\ref{theo:main}~\eqref{it:main:N} and thus concludes the proof of Theorem~\ref{theo:main}.

Finally, with the same arguments, we get
\begin{align*}
    \Pr\left(N(\omega,\mathbf{S})>n\right) &= \Pr\left(D_n(\omega,\mathbf{S}_n)< 1\right)\\
    &= \Pr\left(\bar{H}_0(\omega) > u^{-1}(1-p) \ D_{p,n}(\mathbf{S}_n)\right)\\
    &= u\left(u^{-1}(1-p) \ D_{p,n}(\mathbf{S}_n)\right),
\end{align*}
which yields~\eqref{eq:survival} and proves Corollary~\ref{cor:survival}.

\section{Survival probability of a structure}\label{s:continuum} 

In this section, we consider a mechanical structure $E$ (with arbitrary dimension $d$) globally subjected to cyclic loading, with variable amplitude both in space and time. Each loading cycle is represented by a scalar valued field $\{S_i(x), \, x \in E\}$, such that $S_i(x)>0$ is the severity of the $i$-th cycle locally undergone around the point $x$ of the structure. Following the standard approach of continuum mechanics, the structure is considered as the union of disjoint elementary volumes $F^\eps_k \subset E$, $1 \leq k \leq K^\eps$, with maximal size at most equal to $\eps$ (this notion is formalized in Section~\ref{sec:formalisation} below). Each elementary volume $F^\eps_k$ is assumed to behave as a test specimen, subjected to the sequence of severities $\mathbf{S}^{F^\eps_k} = \left( S_i^{F^\eps_k} \right)_{i \geq 1} \in \mathcal{X}$, where $S_i^{F^\eps_k}$ is the value of $S_i(x)$ on $F^\eps_k$, which is assumed to be approximately uniform. For $\omega \in \Omega$, each element $F^\eps_k$ is associated with its random NCF $N^{F^\eps_k}(\omega,\mathbf{S}^{F^\eps_k})$, and the random NCF of the structure is then defined by
\begin{equation}\label{eq:NCF-struct}
    N(\omega,\{\mathbf{S}(x), \, x \in E\}) := \lim_{\eps \to 0} \, \min_{1 \leq k \leq K^\eps} N^{F^\eps_k}\left(\omega,\mathbf{S}^{F^\eps_k}\right).
\end{equation}
Taking the minimum over $1 \leq k \leq K^\eps$ amounts to assuming that the structure fails as soon as one of the elementary volumes fails. We additionally consider the limit $\eps \to 0$, i.e. the limit where elementary volumes become infinitesimally small.

The aim of this section is to derive integral formulas for the survival probability of the structure. In Section~\ref{ss:volume}, we formulate modelling assumptions on the NCFs $\left\{ N^{F^\eps_k}(\omega,S), \, 1 \leq k \leq K^\eps \right\}$, in the continuation of Section~\ref{s:health}, from which we define the notion of initial health of the structure and derive size effects. These theoretical (but fundamental) results are applied in Section~\ref{ss:struct-proba-cond} to obtain an integral formula for the survival probability of the structure, for a given sequence $\{\mathbf{S}(x), \, x \in E\}$ of loadings. In Section~\ref{ss:struct-proba-rand}, we briefly discuss how to take into account the possible randomness of the latter sequence. The application of these formulas in the Weibull--Basquin model is presented in Section~\ref{ss:WB-struct}. 

\subsection{Size effects}\label{ss:volume}

\subsubsection{Formalisation and modelling assumptions on the random NCFs} \label{sec:formalisation}

A natural mathematical framework to derive integral formulas is to endow the domain $E$ occupied by the mechanical structure with a $\sigma$-algebra $\mathcal{E}$ and a finite measure $\lambda$, such that, for any $\eps>0$, $\mathcal{E}$ contains a partition of $E$ into elementary volumes $(F^\eps_k)_{1 \leq k \leq K^\eps}$ with maximal size smaller than $\eps$: $\dis \max_{1 \leq k \leq K^\eps} \lambda(F^\eps_k) \leq \eps$. The measure $\lambda$ may be the standard volume measure, or it may be more complicated in order to account for inhomogeneities in the material (see the discussion in Section~\ref{sss:Poisson} below).

For any $F \in \mathcal{E}$ and $(\omega,S) \in \Omega \times (0,+\infty)$, we next postulate the existence of a random variable $N^F(\omega,S)$ which describes the random NCF of the specimen $F$ if it was isolated from the remainder of the structure and subjected to cyclic loading with constant severity $S>0$. Besides Assumptions~\eqref{ass:A1}, \eqref{ass:A2} and~\eqref{ass:A3} from Section~\ref{s:miner} on each variable $N^F(\omega,S)$, we introduce the following modelling assumptions on the family of random variables $\{N^F(\omega,S), \, F \in \mathcal{E}\}$:
\begin{enumerate}[label=(A\arabic*),ref=A\arabic*,start=4]
  \item\label{ass:A4} If $F \cap G = \emptyset$, then $N^F(\omega,S)$ and $N^G(\omega,S)$ are independent and 
  $$
  N^{F \cup G}(\omega,S) = \min \left(N^F(\omega,S),N^G(\omega,S)\right).
  $$
  \item\label{ass:A5} If $\lambda(F)=\lambda(G)$, then $N^F(\omega,S)$ and $N^G(\omega,S)$ share the same law.
\end{enumerate}
Assumption~\eqref{ass:A4} is also known as the \emph{weakest link} assumption in the literature~\cite{Wei51,CasFer09}. Assumption~\eqref{ass:A5} is a statistical homogeneity assumption: it asserts that the law of the NCF of an elementary volume only depends on its measure under $\lambda$. Assumptions~\eqref{ass:A4} and~\eqref{ass:A5} complement Assumptions~\eqref{ass:A1},  \eqref{ass:A2} and~\eqref{ass:A3} as first principles for the overall computation of the survival probability of the structure. They imply size effects for the elementary random NCFs $\{N^F(\omega,S), \, F \in \mathcal{E}\}$.

\begin{prop}[Size effects for the elementary random NCFs]\label{prop:vol-NCF}
    Under Assumptions~\eqref{ass:A4} and~\eqref{ass:A5}, there exists a function $\tilde{g} : \N \times (0,+\infty) \to [0,+\infty)$, which is increasing in its first variable $n$, such that, for any $F \in \mathcal{E}$,
    \begin{equation*}
        \forall (n,S) \in \N \times (0,+\infty), \qquad \Pr\left(N^F(\omega,S)>n\right) = \exp\left(-\lambda(F) \, \tilde{g}(n,S)\right).
    \end{equation*}
\end{prop}

\begin{proof}
    Fix $(n,S) \in \N \times (0,+\infty)$. For any $F \in \mathcal{E}$, Assumption~\eqref{ass:A5} asserts that the quantity $\Pr(N^F(\omega,S)>n)$ only depends on $F$ through $\lambda(F)$, it is therefore denoted by $v(n,S,\lambda(F))$. From Assumption~\eqref{ass:A4}, for any disjoint sets $F,G \in \mathcal{E}$, we have
    \begin{align*}
        v(n,S,\lambda(F)+\lambda(G)) &= \Pr\left(N^{F \cup G}(\omega,S)>n\right)\\
        &= \Pr\left(N^F(\omega,S)>n\right) \, \Pr\left(N^G(\omega,S)>n\right)\\
        &= v(n,S,\lambda(F)) \, v(n,S,\lambda(G)).
    \end{align*}
    The assumption that $\mathcal{E}$ contains partitions of $E$ with arbitrarily small maximal measure implies that $\lambda(F)$ takes a continuum of values~\cite[Corollary~1.12.10, p.~56]{Bog07}, which implies, thanks to the functional relation above, that $v(n,S,\lambda(F))$ is exponential in $\lambda(F)$: there exists $\tilde{g}(n,S)$ such that $v(n,S,\lambda(F))=\exp(-\lambda(F) \, \tilde{g}(n,S))$. Since the quantity $v$ decreases from $1$ to $0$ when $n$ varies in $\N$, we conclude that, as a function of $n$, $\tilde{g}$ increases from $0$ to $+\infty$.
\end{proof}

\subsubsection{Normalisation and size effects on the initial health of the structure}\label{sec:norma}

In view of Proposition~\ref{prop:barH}, to each $F \in \mathcal{E}$, one may associate an initial health $\bar{H}^F_0(\omega)$ and a scale function $\langle N^F \rangle(S)$ such that
\begin{equation*}
    \forall (\omega,S) \in \Omega \times (0,+\infty), \qquad N^F(\omega,S) = \langle N^F\rangle(S) \, \bar{H}^F_0(\omega).
\end{equation*}
Clearly, this identity does not uniquely define the pair $\left(\bar{H}^F_0(\omega),\langle N^F \rangle(S)\right)$, because, for any $c>0$, the pair $\left(\bar{H}^F_0(\omega)/c,c \, \langle N^F \rangle(S)\right)$ is equally valid. As pointed out in Remarks~\ref{rk:norma} and~\ref{rk:norma2}, this is reminiscent of the fact that Assumption~\eqref{ass:A3} does not define the functions $\langle N \rangle$ and $u$ in a unique way. The following result provides a consistent normalisation of the pairs $\left(\bar{H}^F_0(\omega),\langle N^F \rangle(S)\right)$ when $F$ varies in $\mathcal{E}$.

\begin{lem}[Normalisation of $\bar{H}^F_0(\omega)$ and $\langle N^F \rangle(S)$]\label{lem:norm-barH0}
    Let Assumptions~\eqref{ass:A4} and~\eqref{ass:A5} hold. There exists a collection of positive constants $\{c_F, \, F \in \mathcal{E}\}$ such that the function $\langle N^F\rangle'(S) = c_F \, \langle N^F\rangle(S)$ does not depend on $F$.
\end{lem}
\begin{proof}
    Fix $\tilde{g}_0 > 0$ and, for any $F \in \mathcal{E}$, set $p_F = 1 - \exp(-\lambda(F) \, \tilde{g}_0)$. Next, define $\langle N^F\rangle'(S)$ to be the quantile $N^F_{p_F}(S)$ of order $p_F$ of the variable $N^F(\omega,S)$. By Remark~\ref{rk:SN-A3}, $\langle N^F\rangle'(S)$ writes under the form
    \begin{equation*}
        \langle N^F\rangle'(S) = c_F \, \langle N^F\rangle(S), \qquad c_F = (u^F)^{-1}(1-p_F),
    \end{equation*}
    where $u^F$ is the shape function associated with the scale function $\langle N^F\rangle(S)$ by Assumption~\eqref{ass:A3}. Therefore, to complete the proof of the lemma, it suffices to show that $\langle N^F\rangle'(S)$ does not depend on $F$.
    
    By Proposition~\ref{prop:vol-NCF}, we have 
    \begin{equation*}
        \langle N^F\rangle'(S) = N^F_{p_F}(S) = \tilde{g}^{-1}\left(-\frac{\ln(1-p_F)}{\lambda(F)},S\right) = \tilde{g}^{-1}\left(\tilde{g}_0,S\right),
    \end{equation*}
    where $\tilde{g}^{-1}(\cdot,S)$ is the inverse function of $g(\cdot,S)$ and where the last equality stems from the definition of $p_F$. We see that the above right-hand side does not depend on $F$, which completes the proof.
\end{proof}

From now on, we set $\langle N \rangle(S) = c_F \, \langle N^F \rangle(S)$ and simply rewrite $\bar{H}_0^F(\omega)$ for $\bar{H}_0^F(\omega)/c_F$, so that
\begin{equation}\label{eq:renorm-N}
    \forall F \in \mathcal{E}, \quad \forall (\omega,S) \in \Omega \times (0,+\infty), \qquad N^F(\omega,S) = \langle N \rangle(S) \, \bar{H}^F_0(\omega),
\end{equation}
where the first factor is independent on $F$. The derivation of integral formulas for the survival probability of the structure relies on the following statement concerning the initial health.

\begin{theo}[Size effects of the initial health]\label{theo:vol-H0}
    With the normalisation of Lemma~\ref{lem:norm-barH0}, the random variables $\left\{\bar{H}^F_0(\omega), \, F \in \mathcal{E} \right\}$ are such that:
    \begin{enumerate}[label=(\roman*),ref=\roman*]
        \item\label{it:norm-barH0:1} if $F \cap G = \emptyset$, then $\bar{H}^F_0(\omega)$ and $\bar{H}^G_0(\omega)$ are independent and $\bar{H}^{F \cup G}_0(\omega) = \min\left(\bar{H}^F_0(\omega),\bar{H}^G_0(\omega)\right)$;
        \item\label{it:norm-barH0:2} if $\lambda(F)=\lambda(G)$, then $\bar{H}^F_0(\omega)$ and $\bar{H}^G_0(\omega)$ share the same law.
    \end{enumerate}
    Moreover, there exists a function $g : [0,+\infty) \to [0,+\infty)$ which increases from $0$ to $+\infty$ such that
    \begin{equation} \label{eq:indy2}
        \forall F \in \mathcal{E}, \quad \forall h \geq 0, \qquad \Pr\left(\bar{H}^F_0(\omega)>h\right) = \exp\left(-\lambda(F) \, g(h)\right).
    \end{equation}
\end{theo}

\begin{proof}
    Assertions~(i) and~(ii) follow from~\eqref{eq:renorm-N} and Assumptions~\eqref{ass:A4} and~\eqref{ass:A5}. We now prove~\eqref{eq:indy2}. Using~\eqref{eq:renorm-N} and Proposition~\ref{prop:vol-NCF}, we have, for any $h \geq 0$ and $S>0$,
    \begin{equation*}
        \Pr\left(\bar{H}^F_0(\omega)>h\right) = \Pr\left(N^F(\omega,S)>h \, \langle N \rangle(S)\right) = \exp\big(-\lambda(F) \, \tilde{g}(h \, \langle N \rangle(S),S)\big).
    \end{equation*}
    We observe that the left-hand side of the above equality does not depend on $S$. The right-hand side hence does not depend on $S$, which shows~\eqref{eq:indy2} with the function $g(h)=\tilde{g}(h \, \langle N \rangle(S),S)$, which depends neither on $F$ nor on $S$.
\end{proof}

\begin{rk}\label{rk:N-g-identif}
    The functions $\langle N \rangle(S)$ and $g(h)$ appearing in~\eqref{eq:indy2} are related with the NCF $N^F(\omega,S)$ by the identity
    \begin{equation} \label{eq:de_la_rk:N-g-identif}
        \forall n \geq 0, \qquad \Pr\left(N^F(\omega,S)>n\right) = \Pr\left(\bar{H}^F_0(\omega) > \frac{n}{\langle N\rangle(S)} \right) = \exp\left(-\lambda(F) \, g\left(\frac{n}{\langle N\rangle(S)}\right)\right),
    \end{equation}
    where the first equality stems from~\eqref{eq:renorm-N}.

    Therefore, if experimental results are available and yield the family of probability distributions $\{\mathrm{P}^F_S, \, S>0\}$ associated with a test specimen $F$ with similar mechanical properties as the structure, the functions $\langle N \rangle(S)$ and $g(h)$ can be identified, up to a linear transform of the form $(\langle N \rangle(S),g(h)) \to (c \, \langle N \rangle(S),g(c \, h))$ for some $c>0$.
\end{rk}

\subsubsection{A consistent microscopic model}\label{sss:Poisson}

Independently from the first-principle derivation performed above, the existence of a family of positive random variables $\dis \left\{ \bar{H}^F_0(\omega), \, F \in \mathcal{E} \right\}$ satisfying the conclusions of Theorem~\ref{theo:vol-H0} can be proved by a direct microscopic construction. In the present section, we develop this construction, which provides a natural interpretation of the variable $\bar{H}^F_0(\omega)$ as well as possible extensions of the present work, for example to take into account different types of flaws. However, we emphasize that the contents of the next sections, dedicated to the survival probability of structures, only depend on the statement of Theorem~\ref{theo:vol-H0} and not on this microscopic model.

Fix a function $g : [0,+\infty) \to [0,+\infty)$ increasing from $0$ to $+\infty$, and consider a realisation 
\begin{equation*}
    \{(X_\ell(\omega), \mathrm{H}_\ell(\omega)), \, \ell \geq 1\}
\end{equation*}
of a Poisson point process on the product space $E \times (0,+\infty)$, with intensity measure $\lambda \otimes \dd g$, where $\dd g$ is the measure on $(0,+\infty)$ whose cumulative distribution function is $g$. For any $F \in \mathcal{E}$, set
\begin{equation*}
    \bar{H}_0^F(\omega) = \min_{\ell \text{ s.t. } X_\ell(\omega) \in F} \mathrm{H}_\ell(\omega).
\end{equation*}
Then it follows from standard properties of Poisson point processes that the family of random variables $\dis \left\{ \bar{H}^F_0(\omega), \, F \in \mathcal{E} \right\}$ satisfies the conclusions of Theorem~\ref{theo:vol-H0}. 

Assuming that fatigue originates in the propagation of microscopic flaws, the interpretation of this construction is that these flaws are initially located at the random points $X_\ell(\omega)$ and associated with a random initial microscopic health $\mathrm{H}_\ell(\omega)$. Consistently with the weakest-link principle, the health of a volume $F$ is the smallest microscopic health of the flaws located in $F$. With this construction, we recover in particular the fact that, for a given activation level $H>0$, the set of points $X_\ell(\omega)$ such that $\mathrm{H}_\ell(\omega) \leq H$ forms a Poisson point process on $E$, with intensity measure $g(H) \, \lambda$. This approach is reminiscent of weakest-link models of failure for brittle materials, where microscopic flaws which are activated at a certain stress level are assumed to be randomly distributed in the material according to a Poisson point process. In the case of the Weibull model, this point process has an intensity proportional to $H^m$ (see~\cite[Section~20.3.3]{Jeu21}).

This interpretation also highlights the role of the measure $\lambda$: it allows to model the statistical distribution of flaws, and may be taken to be non-uniform if flaws are known to be more concentrated in certain parts of the structure. In fact, $\lambda$ need not be related with the volume measure, and may be for instance a surface measure if flaws are only present at the surface of the structure.

A natural generalisation of this approach would then be to model flaws of different types, for instance in volume {\em and} in surface, as is addressed in~\cite{BomMaySch97}. The corresponding microscopic model would be to superpose two independent Poisson point processes $\dis \left\{(X^1_\ell(\omega), \mathrm{H}^1_\ell(\omega)), \, \ell \geq 1\right\}$ and $\dis \left\{(X^2_\ell(\omega), \mathrm{H}^2_\ell(\omega)), \, \ell \geq 1\right\}$, with respective intensity measures $\lambda^1 \otimes \dd g^1$ and $\lambda^2 \otimes \dd g^2$, where $\lambda^1$ and $\lambda^2$ are the volume and surface measures, respectively.

\subsection{Survival probability conditionally on loading}\label{ss:struct-proba-cond}

In this section, we compute the survival probability of the whole structure, in the case where it is subjected to cyclic loading represented by the sequence $\{\mathbf{S}(x), \, x \in E\}$ which is assumed to be deterministic. Fix $\eps>0$ and let $F^\eps_k \in \mathcal{E}$, $1 \leq k \leq K^\eps$, be a partition of the structure $E$ such that the size of each element $F^\eps_k$ is smaller than $\eps$. Assuming that, for any cycle $i \geq 1$, the field $S_i(x)$ is approximately uniform and equal to some value $S_i^{F^\eps_k}$ in $F^\eps_k$, the random NCF $N^{F^\eps_k}(\omega,\mathbf{S}^{F^\eps_k})$ can be defined as in Section~\ref{s:miner}. Using Proposition~\ref{prop:barH} and Theorem~\ref{theo:vol-H0}, we get, for any $n \geq 1$, 
\begin{align*}
    \Pr\left(\min_{1 \leq k \leq K^\eps} N^{F^\eps_k}\left(\omega,\mathbf{S}^{F^\eps_k}\right) > n\right) &= \prod_{k=1}^{K^\eps} \Pr\Big(N^{F^\eps_k}\left(\omega,\mathbf{S}^{F^\eps_k}\right) > n\Big)\\
    &= \prod_{k=1}^{K^\eps} \Pr\left(\bar{H}^{F^\eps_k}_0(\omega) > \sum_{i=1}^n \frac{1}{\langle N \rangle(S^{F^\eps_k}_i)}\right)\\
    &= \exp\left(-\sum_{k=1}^{K^\eps} \lambda(F^\eps_k) \, g\left(\sum_{i=1}^n \frac{1}{\langle N \rangle(S^{F^\eps_k}_i)}\right)\right).
\end{align*}
When $\eps$ becomes small, the sum in the exponential converges to an integral\footnote{This statement can be made rigorous if the field $\{S_i(x), \, x \in E\}$ is continuous, and if the elementary volumes $F^\eps_k$ have a diameter which vanishes uniformly in $k$ when $\eps \to 0$.}, and we deduce the following continuum formula for the survival probability.

\begin{theo}[Survival probability]\label{theo:survival-struct}
    Under Assumptions~\eqref{ass:A1} to~\eqref{ass:A5}, the random NCF $N(\omega,\{\mathbf{S}(x), \, x \in E\})$ of the structure, which is defined by~\eqref{eq:NCF-struct}, satisfies, for any $n \geq 0$,
    \begin{equation*}
        \Pr\left(N(\omega,\{\mathbf{S}(x), \, x \in E\})>n\right) = \exp\left(-\int_{x \in E} g\left(\sum_{i=1}^n \frac{1}{\langle N \rangle(S_i(x))}\right) \, \lambda(\dd x)\right),
    \end{equation*}
    where the scale function $S \mapsto \langle N \rangle(S)$ is defined by Lemma~\ref{lem:norm-barH0} and the function $h \mapsto g(h)$ is provided in Theorem~\ref{theo:vol-H0}.
\end{theo}

In view of Remark~\ref{rk:N-g-identif}, the quantity $\dis g\left(\sum_{i=1}^n 1/\langle N \rangle(S_i(x))\right)$, which does not change if the pair $(\langle N \rangle(S),g(h))$ is replaced by $(c \, \langle N \rangle(S),g(c \, h))$ for some $c>0$, is identifiable if experimental results are available for a specimen with similar mechanical properties as the structure. Thus, from the practical point of view, Theorem~\ref{theo:survival-struct} shows how to use experimental results for specimens subjected to cyclic loading with constant severity in order to compute the survival probability of a structure subjected to cyclic loading with variable in time and non-uniform in space severity.

\subsection{Survival probability with random loading}\label{ss:struct-proba-rand} 

In this section, we introduce a formalism to model the cases when the loading sequence $\{\mathbf{S}(x), \, x \in E\}$ is random. We assume that this sequence is independent from the initial state (and in particular the initial health) of the structure. To proceed, we assume that, for any $i \geq 1$, the scalar field $\{S_i(x), \, x \in E\}$ is a random variable in some functional space $\mathcal{F}(E)$, and denote $\Pr^*$ the probability measure on $\Omega \times \mathcal{F}(E)^\N$ which is the product of $\Pr$ on $\Omega$ and of the law of $\{\mathbf{S}(x), \, x \in E\}$ on $\mathcal{F}(E)^\N$. A direct conditioning argument shows that, in this context, Theorem~\ref{theo:survival-struct} yields the identity
\begin{equation} \label{eq:proba_random_loading}
    \Pr^*\left(N(\omega,\{\mathbf{S}(x), \, x \in E\})>n\right) = \Exp^*\left[\exp\left(-\int_{x \in E} g\left(\sum_{i=1}^n \frac{1}{\langle N \rangle(S_i(x))}\right) \, \lambda(\dd x)\right)\right]
\end{equation}
for the survival probability of the structure. This formula opens the door to analytical calculations if an explicit statistical model is specified for the law of $\{\mathbf{S}(x), \, x \in E\}$, or to Monte Carlo methods if sampling from $\{\mathbf{S}(x), \, x \in E\}$ is possible. An example of the latter case is detailed in Section~\ref{sss:WB-MC}.

\subsection{The Basquin--Weibull model}\label{ss:WB-struct} 

In this section, the survival probability for the structure given by Theorem~\ref{theo:survival-struct} is implemented in the case where experimental results $\dis \left\{\mathrm{P}^F_S, \, S>0 \right\}$ are described by the Weibull--Basquin model from Section~\ref{ss:weibull-basquin}, for some parameters $m$ and $\alpha$ that are assumed to be valid for the whole structure. 

\medskip

In this case, we first follow the identification procedure performed in Proposition~\ref{prop:vol-NCF}, Lemma~\ref{lem:norm-barH0} and Theorem~\ref{theo:vol-H0}. Consider a test specimen $F$. Assuming a Weibull model for that specimen, we have (see~\eqref{eq:weibull}) that
\begin{equation*}
    \forall n \geq 0, \qquad \Pr(N^F(\omega,S)>n) = \exp\left(-\left(\frac{n}{\langle N^F \rangle(S)}\right)^m\right),
\end{equation*}
for some scale function $\langle N^F \rangle$. In view of Proposition~\ref{prop:vol-NCF}, we can write that
$$
\lambda(F) \, \tilde{g}(n,S) = \left(\frac{n}{\langle N^F \rangle(S)}\right)^m,
$$
where we recall that the function $\tilde{g}$ is independent of $F$. As in the proof of Lemma~\ref{lem:norm-barH0}, we next fix $\tilde{g}_0 > 0$, set $p_F = 1 - \exp(-\lambda(F) \, \tilde{g}_0)$ and introduce $c_F = u^{-1}(1-p_F)$, where $u$ is the shape function associated to the specimen $F$, which reads $u(h)=\exp(-h^m)$. We thus compute that $c_F^m = \lambda(F) \, \tilde{g}_0$, and therefore obtain that
$$
\tilde{g}(n,S) = \tilde{g}_0 \, \left(\frac{n}{c_F \, \langle N^F \rangle(S)}\right)^m = \tilde{g}_0 \, \left(\frac{n}{\langle N \rangle(S)}\right)^m,
$$
where $\langle N \rangle(S)$ is independent of $F$. Following the proof of Theorem~\ref{theo:vol-H0}, we eventually obtain that
\begin{equation} \label{eq:g_pour_weibull}
g(h) = \tilde{g}(h \, \langle N \rangle(S),S) = \tilde{g}_0 \, h^m,
\end{equation}
which is indeed independent of $F$ and $S$.

\subsubsection{Transcription of Theorem~\ref{theo:survival-struct}} 

Let us denote by $\kappa_\mathrm{ref}$ the constant $\kappa(m,\alpha,p)$ introduced in Section~\ref{ss:weibull-basquin} for the reference test specimen $F$, and recall that the value of this quantity is assumed to be known from experimental test --- in fact, only $\alpha$, $m$ and the detail category associated with the S-N curve for the reference probability $p$ need to be given. We also denote by $\lambda_\mathrm{ref} = \lambda(F)$ the size of the specimen $F$. It then follows from the above computations that, for any $n$ and $S$,
\begin{equation}\label{eq:g-N-WB}
    g\left(\frac{n}{\langle N\rangle(S)}\right) = \tilde{g}_0 \, \left(\frac{n}{\langle N\rangle(S)}\right)^m = \frac{1}{\lambda_\mathrm{ref}} \left(\frac{n}{\langle N^F \rangle(S)}\right)^m = \frac{1}{\lambda_\mathrm{ref}}\left(\frac{n \, S^\alpha}{\kappa_\mathrm{ref}}\right)^m,
\end{equation}
where the last relation stems from~\eqref{eq:def_kappa}, written for the specimen $F$. We also note that~\eqref{eq:g-N-WB} can be obtained by identifying the equation~\eqref{eq:Basquin:3} (for a sequence of constant severities) with the equation~\eqref{eq:de_la_rk:N-g-identif}.

In view of Theorem~\ref{theo:survival-struct}, the survival probability for the structure therefore writes
\begin{equation}\label{eq:WB-surv:1}
    \forall n \geq 0, \qquad \Pr\left(N(\omega,\{\mathbf{S}(x), \, x \in E\})>n\right) = \exp\left(-\frac{1}{\lambda_\mathrm{ref}}\int_{x \in E} \left(\sum_{i=1}^n \frac{S_i(x)^\alpha}{\kappa_\mathrm{ref}}\right)^m \, \lambda(\dd x)\right).
\end{equation}

\subsubsection{Linear elasticity} 

If the structure is assumed to be linearly elastic, then the severity field $\{S_i(x), \, x \in E\}$ of the $i$-th cycle rewrites under the form
\begin{equation*}
    S_i(x) = P_i \, \Sevu(x),
\end{equation*}
where $\{\Sevu(x), \, x \in E\}$ describes the \emph{unitary severity} of the material and $P_i > 0$ is the global load of the $i$-th cycle. In this case, the random NCF of the structure $N(\omega,\{\mathbf{S}(x), \, x \in E\})$ only depends on $\{\mathbf{S}(x), \, x \in E\}$ through the sequence $(P_i)_{i \geq 1}$, and we denote it by $N(\omega,(P_i)_{i \geq 1})$. The formula~\eqref{eq:WB-surv:1} rewrites
\begin{equation}\label{eq:indy3}
    \forall n \geq 0, \qquad \Pr\left(N(\omega,(P_i)_{i \geq 1})>n\right) = \exp\left(-\mQ \left(\sum_{i=1}^n P_i^\alpha\right)^m\right),
\end{equation}
with a constant (with respect to the sequence $(P_i)_{i \geq 1}$) term
\begin{equation} \label{eq:indy4}
    \mQ = \frac{1}{\lambda_\mathrm{ref}} \int_{x \in E} \left(\frac{\Sevu(x)^\alpha}{\kappa_\mathrm{ref}}\right)^m \ \lambda(\dd x).
\end{equation}

The evaluation of $\mQ$ requires a mechanical computation, either analytical or by numerical methods (e.g., finite element methods), of the unitary severity field in the whole structure $E$. This may be costly or not possible. However, as soon as the exponent $\alpha m$ is large, it may be expected that only terms for which $\Sevu(x)$ is large will contribute to the overall value of the integral $\mQ$. It is therefore only necessary to know $\Sevu(x)$ in the neighbourhood of such maximal severity points, which shall be referred to as \emph{hot points}. In this perspective, a precise approximation procedure, based on the use of the Laplace method, is detailed in Appendix~\ref{app:laplace}.

We note that, once $\mQ$ is computed, evaluating the survival probability of the structure (either by~\eqref{eq:indy3} in the case of deterministic loadings, or by~\eqref{eq:indy5} below in the case of random loadings) is inexpensive.

\subsubsection{Monte Carlo method in the random loading case}\label{sss:WB-MC}

To complete the presentation of the Weibull--Basquin case, we assume that the loading is random and in the linear elasticity regime, and that the sequence of global loads $(P_i)_{i \geq 1}$ are independent and identically distributed according to some measure under which sampling is possible. Then, for any fixed value of $n$, and in view of~\eqref{eq:proba_random_loading}, the probability that the NCF of the structure be larger than $n$ can be estimated, for large $R$, by
\begin{equation} \label{eq:indy5}
    \hat{P}^*_{R,n} = \frac{1}{R} \sum_{r=1}^R \exp\left(-\mQ \left(\sum_{i=1}^n P_{i,r}^\alpha\right)^m\right),
\end{equation}
where the array $(P_{i,r})_{1 \leq i \leq n, \, 1 \leq r \leq R}$ contains independent realisations of the global load.

\subsubsection{Density probability function of the failure point} \label{sss:fail} 

In the Weibull--Basquin model, the linear elasticity assumption enables a factorisation of the scale function $\langle N \rangle(S)$ which provides, in addition to the law of the NCF, the probability density function of the point $x \in E$ at which failure occurs. Indeed, it follows from~\eqref{eq:g-N-WB} that there exists a constant $c>0$ (whose precise value is $\dis c = (\tilde{g}_0 \, \lambda_\mathrm{ref})^{1/m} \, \kappa_\mathrm{ref}$) such that, for any $S>0$,
\begin{equation} \label{eq:def_g_WB_1}
    \langle N \rangle(S) = c \, S^{-\alpha}.
\end{equation}
We then deduce from~\eqref{eq:g_pour_weibull} that, for any $h \geq 0$,
\begin{equation} \label{eq:def_g_WB_2}
    g(h) = \tilde{g}_0 \, h^m = \frac{1}{\lambda_\mathrm{ref}}\left(\frac{c \, h}{\kappa_\mathrm{ref}}\right)^m.
\end{equation}
Therefore, taking a partition $\dis \left\{F^\eps_k, \, 1 \leq k \leq K^\eps \right\}$ of the structure $E$ with maximal measure $\eps$, and denoting by $(\Sevu)^{F^\eps_k}$ the average value of the unitary severity in $F^\eps_k$, we infer from~\eqref{eq:evol_barH}, \eqref{eq:def_g_WB_1} and the linear elasticity assumption that
\begin{equation*}
    \bar{H}^{F^\eps_k}_{n+1}\left(\omega,\mathbf{S}^{F^\eps_k}_{n+1}\right) = \bar{H}^{F^\eps_k}_n\left(\omega,\mathbf{S}^{F^\eps_k}_n\right) - \frac{1}{c} \, P_{n+1}^\alpha \, \left((\Sevu)^{F^\eps_k}\right)^\alpha.
\end{equation*}
As a consequence, the quantities $\dis \bar{\mathcal{H}}^{F^\eps_k}_{n,{\rm norm}}\left(\omega,\mathbf{S}^{F^\eps_k}_n\right)$ defined by
\begin{equation*}
    \bar{\mathcal{H}}^{F^\eps_k}_{n,{\rm norm}}\left(\omega,\mathbf{S}^{F^\eps_k}_n\right) = \frac{\bar{H}^{F^\eps_k}_n(\omega,\mathbf{S}^{F^\eps_k}_n)}{((\Sevu)^{F^\eps_k})^\alpha/c}
\end{equation*}
all decrease by the same quantity $P_n^\alpha$, independently from $k$, the index of the element in the partition. Since
\begin{equation*}
    N^{F^\eps_k}\left(\omega, \mathbf{S}^{F^\eps_k}\right) = \inf\left\{n \geq 1: \, \bar{H}^{F^\eps_k}_n\left(\omega,\mathbf{S}^{F^\eps_k}_n\right) \leq 0 \right\} = \inf\left\{n \geq 1: \, \bar{\mathcal{H}}^{F^\eps_k}_{n,{\rm norm}}\left(\omega,\mathbf{S}^{F^\eps_k}_n\right) \leq 0 \right\},
\end{equation*}
we deduce that the first element to fail is the element $k$ for which $\bar{\mathcal{H}}^{F^\eps_k}_{0,{\rm norm}}(\omega)$ takes the smallest value. As a consequence, for a given index $k_0$,
\begin{align*}
    \Pr(\text{the element $k_0$ fails first}) &= \Pr\left(\forall k \not= k_0, \ \ \bar{\mathcal{H}}^{F^\eps_k}_{0,{\rm norm}}(\omega) > \bar{\mathcal{H}}^{F^\eps_{k_0}}_{0,{\rm norm}}(\omega)\right) \\
    &=\Pr\left(\forall k \not= k_0, \ \ \frac{\bar{H}^{F^\eps_k}_0(\omega)}{((\Sevu)^{F^\eps_k})^\alpha/c} > \frac{\bar{H}^{F^\eps_{k_0}}_0(\omega)}{((\Sevu)^{F^\eps_{k_0}})^\alpha/c}\right)\\
    &=\Exp\left[\prod_{k \not=k_0} \Pr\left(\frac{\bar{H}^{F^\eps_k}_0(\omega)}{((\Sevu)^{F^\eps_k})^\alpha} > \frac{\bar{H}^{F^\eps_{k_0}}_0(\omega)}{((\Sevu)^{F^\eps_{k_0}})^\alpha} \Bigg|\bar{H}^{F^\eps_{k_0}}_0(\omega)\right)\right]\\
    &=\Exp\left[\prod_{k \not=k_0} \exp\left(-\lambda(F^\eps_k) \ g\left(\frac{((\Sevu)^{F^\eps_k})^\alpha}{((\Sevu)^{F^\eps_{k_0}})^\alpha} \ \bar{H}^{F^\eps_{k_0}}_0(\omega)\right)\right)\right],
\end{align*}    
where we have used Theorem~\ref{theo:vol-H0} (at the third line for the independence of the $\left\{ \bar{H}^{F^\eps_k}_0(\omega) \right\}_k$ and at the last line with~\eqref{eq:indy2}). Using now the expression~\eqref{eq:def_g_WB_2} of $g$, we deduce that
\begin{equation} \label{eq:calcul_inter}
\Pr(\text{the element $k_0$ fails first})
=
\Exp\left[\prod_{k \not=k_0} \exp\left(-\frac{\lambda(F^\eps_k)}{\lambda_\mathrm{ref}} \left(\frac{c}{\kappa_\mathrm{ref}} \, \frac{((\Sevu)^{F^\eps_k})^\alpha}{((\Sevu)^{F^\eps_{k_0}})^\alpha} \ \bar{H}^{F^\eps_{k_0}}_0(\omega)\right)^m\right)\right].
\end{equation}
Writting~\eqref{eq:indy2} for the element $F^\eps_{k_0}$ and using the expression~\eqref{eq:def_g_WB_2} of $g$, we see that
$$
\forall h \geq 0, \qquad \Pr\left(\bar{H}^{F^\eps_{k_0}}_0(\omega)>h\right) = \exp\left(- \frac{\lambda(F^\eps_{k_0})}{\lambda_\mathrm{ref}}\left(\frac{c \, h}{\kappa_\mathrm{ref}}\right)^m\right),
$$
and we thus deduce that the random variable
\begin{equation*}
    Z_{k_0}(\omega) = \left(\frac{\lambda(F^\eps_{k_0})}{\lambda_\mathrm{ref}}\right)^{1/m} \frac{c}{\kappa_\mathrm{ref}} \ \bar{H}^{F^\eps_{k_0}}_0(\omega)
\end{equation*}
is a standard Weibull variable, namely such that $\Pr(Z_{k_0}(\omega)>z)=\exp(-z^m)$. For such variables, a direct computation shows that, for any $a>0$,
\begin{equation*}
    \Exp\left[\exp(-a \, Z_{k_0}^m(\omega))\right] = \frac{1}{1+a}.
\end{equation*}
We therefore infer from~\eqref{eq:calcul_inter} that
\begin{align*}
    \Pr(\text{the element $k_0$ fails first})
    &= \Exp\left[ \exp\left(-\sum_{k \not=k_0} \frac{\lambda(F^\eps_k)}{\lambda(F^\eps_{k_0})} \left(\frac{((\Sevu)^{F^\eps_k})^\alpha}{((\Sevu)^{F^\eps_{k_0}})^\alpha} \ Z_{k_0}(\omega)\right)^m\right)\right]\\
    &= \Exp\left[ \exp\left(-\left(\sum_{k \not= k_0} \frac{\lambda(F^\eps_k)}{\lambda(F^\eps_{k_0})} \ \frac{((\Sevu)^{F^\eps_k})^{\alpha m}}{((\Sevu)^{F^\eps_{k_0}})^{\alpha m}}\right) Z_{k_0}^m(\omega)\right)\right]\\
    &= \frac{1}{\dis 1+\sum_{k \not= k_0} \frac{\lambda(F^\eps_k)}{\lambda(F^\eps_{k_0})} \ \frac{((\Sevu)^{F^\eps_k})^{\alpha m}}{((\Sevu)^{F^\eps_{k_0}})^{\alpha m}}}\\
    &= \frac{\lambda(F^\eps_{k_0}) \ ((\Sevu)^{F^\eps_{k_0}})^{\alpha m}}{\sum_k \lambda(F^\eps_k) \ ((\Sevu)^{F^\eps_k})^{\alpha m}}.
\end{align*}
Taking the limit $\eps \to 0$, we conclude that the law of the failure point has density
\begin{equation} \label{eq:failure_density}
    p_\mathrm{fail}(x) = \frac{\Sevu(x)^{\alpha m}}{\dis \int_{x' \in E} \Sevu(x')^{\alpha m} \,  \lambda(\dd x')}
\end{equation}
with respect to the measure $\lambda(\dd x)$. This law depends on the geometry and the mechanical parameters of the structure through the map $\Sevu$, and on the local initial health through the measure $\lambda$.

\section{Application to the fatigue of an I-steel beam in the Weibull--Basquin model}\label{s:poutre}

In this section, we consider an I-steel beam (which is assumed to be part of a structure, such as a bridge) which undergoes the passage of heavy vehicles. Each passage corresponds to a loading cycle, the severity of which is defined by the maximal Von Mises stress induced by the vehicle.

We work under the Weibull--Basquin model, in the linear elastic regime, and therefore use the formulas of Section~\ref{ss:WB-struct}. In particular, the global load $P_i$ of the $i$-th cycle is the weight of the vehicle. The sequence $(P_i)_{i \geq 1}$ is assumed to be random. We thus infer from~\eqref{eq:indy3} (see also~\eqref{eq:proba_random_loading}) that
\begin{equation}\label{eq:survival-example}
    \Pr^*\left(N(\omega,(P_i)_{i \geq 1})>n\right) = \Exp^*\left[\exp\left(-\mQ \left(\sum_{i=1}^n P_i^\alpha\right)^m\right)\right],
\end{equation}
with
\begin{equation*}
    \mQ = \frac{1}{\lambda_\mathrm{ref}}\int_{x \in E} \left(\frac{\Sevu(x)^\alpha}{\kappa_\mathrm{ref}}\right)^m \lambda(\dd x),
\end{equation*}
for the survival probability of the beam.

The geometry of the beam $E$ is described in Section~\ref{ss:poutre-geo}. The unitary severity field $\Sevu(x)$ is defined and computed analytically in Section~\ref{ss:poutre-sev}. The parameters $\alpha$, $m$, $\lambda_\mathrm{ref}$ and $\kappa_\mathrm{ref}$ are specified in Section~\ref{ss:poutre-ref}. The numerical experiments results are presented in Section~\ref{ss:poutre-exp}, and the density of the failure point is briefly discussed in Section~\ref{ss:poutre-fail}.

\subsection{Geometry of the beam}\label{ss:poutre-geo} 

A representation of the beam, and the numerical values taken for the geometric parameters in the numerical application, are provided in Figure~\ref{fig:poutre}. The moment of inertia along the $z$ coordinate writes
\begin{equation*}
    I_z = \frac{b \, e^3}{12} + \frac{b \, e \, (h-e)^2}{2} + \frac{f \, (h-2e)^3}{12}.
\end{equation*}

\begin{figure}[htpb]
    \begin{minipage}{.55\textwidth}
        \includegraphics[width=\textwidth]{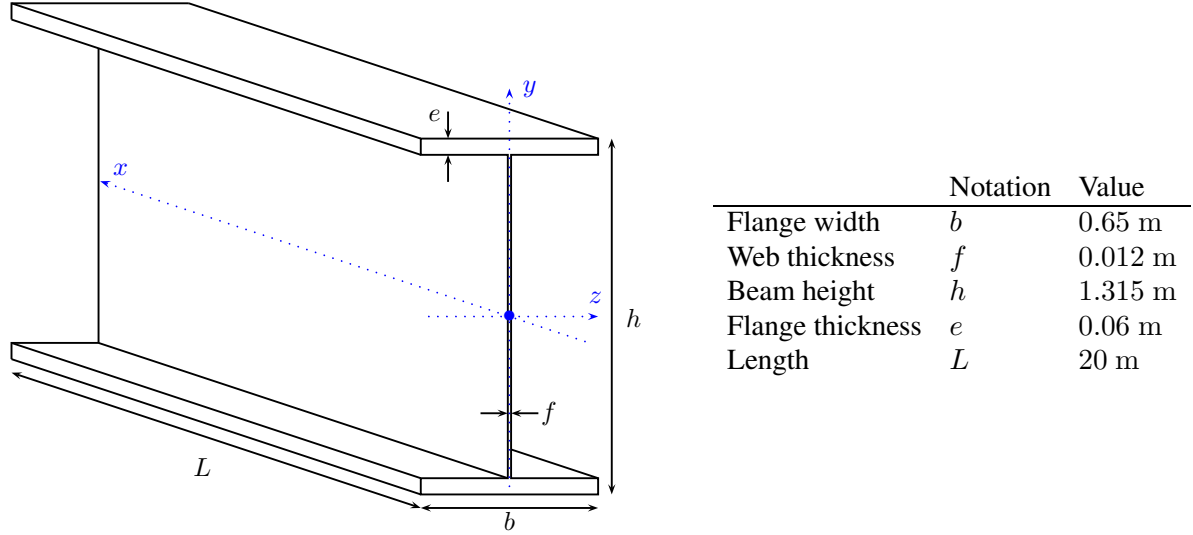}
    \end{minipage}\begin{minipage}{.45\textwidth}
        \centering
        \begin{tabular}{lll}
            & Notation & Value\\
            \hline
            Flange width & $b$ & $0.65~\mathrm{m}$\\
            Web thickness & $f$ & $0.012~\mathrm{m}$\\
            Beam height & $h$ & $1.315~\mathrm{m}$\\
            Flange thickness & $e$ & $0.06~\mathrm{m}$\\
            Length & $L$ & $20~\mathrm{m}$
        \end{tabular}
    \end{minipage}
    \caption{Schematic representation of the I-beam and numerical values of the geometric parameters.} 
    \label{fig:poutre}
\end{figure}

\subsection{Unitary severity}\label{ss:poutre-sev}

For a vehicle with weight $P$, expressed in $\mathrm{MN}$, and located in $x=a \in [0,L]$, $y=h/2$, $z=0$, the bending moment at $x$ writes
\begin{equation*}
    M(x,a) = P \, M^u(x,a), \qquad M^u(x,a) = \begin{cases}
        \frac{(L-a)x}{L}, & 0 \leq x \leq a,\\
        \frac{(L-x)a}{L}, & a < x \leq L,
    \end{cases}
\end{equation*}
and the shear force at $x$ is given by
\begin{equation*}
    V(x,a) = P \, V^u(x,a), \qquad V^u(x,a) = \begin{cases}
        -\frac{L-a}{L}, & 0 \leq x \leq a,\\
        \frac{a}{L}, & a < x \leq L.
    \end{cases}
\end{equation*}
The stress tensor at a point $(x,y,z) \in E$ then writes
\begin{equation*}
    \sigma(x,y,z;a) = P \, \sigma^u(x,y,z;a), \qquad \sigma^u(x,y,z;a) = \begin{pmatrix}
        \sigma^u_{xx} & \sigma^u_{xy} & \sigma^u_{xz}\\
        \sigma^u_{xy} & 0 & 0\\
        \sigma^u_{xz} & 0 & 0
    \end{pmatrix},
\end{equation*}
with
\begin{align*}
    \sigma^u_{xx} &= \frac{y \, M^u(x,a)}{I_z},\\
    \sigma^u_{xy} &= \frac{3V^u(x,a)}{2f} \ \frac{bh^2-b(h-2e)^2+f(h-2e)^2-4fy^2}{bh^3-b(h-2e)^3+f(h-2e)^3},\\
    \sigma^u_{xz} &= \frac{3V^u(x,a)}{2e} \ \frac{h^2-(h-2e)^2}{bh^3-b(h-2e)^3+f(h-2e)^3}\left(\frac{b}{2}-|z|\right).
\end{align*}
The Von Mises equivalent unitary stress is defined by
\begin{equation*}
    \sigma^u_\mathrm{VM} = \sqrt{(\sigma^u_{xx})^2 + 3(\sigma^u_{xy} + \sigma^u_{xz})^2},
\end{equation*}
so that the unitary severity finally writes
\begin{equation*}
    \Sevu(x,y,z) = \max_{a \in [0,L]} \sigma^u_\mathrm{VM}(x,y,z;a).
\end{equation*}
The graph of the function $x \in [0,L] \mapsto \Sevu(x,y,0)$, for several values of $y$ from $0$ to $h/2$, is shown on Figure~\ref{fig:sev-u-x}.

\begin{figure}[htpb]
    \centering
    \includegraphics[width=.8\textwidth]{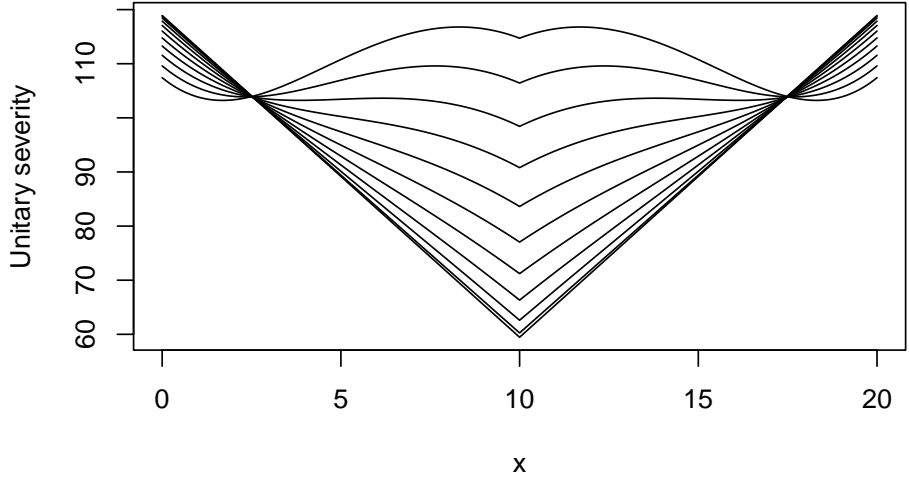}
\caption{Unitary severity $\Sevu(x,y,0)$ as a function of $x \in [0,L]$. Each curve corresponds to a fixed value of $y \in [0,h/2]$. The smaller the value at $x=L/2$, the smaller $y$.} \label{fig:sev-u-x}
\end{figure}

\subsection{Reference parameters}\label{ss:poutre-ref}

The measure $\lambda$ is taken to be the volume measure. The parameters of the material are $\alpha=3$ and $m=1.5$. We assume that experimental data are available for a test specimen with volume $\lambda_\mathrm{ref} = 3 \cdot 10^{-5}~\mathrm{m}^3$, for which the detail category associated with $p = 0.05$ and $\mathsf{N}_p = 2 \cdot 10^6$ cycles is $\mathsf{S}_p = 200~\mathrm{MPa}$. This set of parameters is the same as for the S-N curve shown on Figure~\ref{fig:SN-simu} (see Section~\ref{ss:weibull-basquin}).

\subsection{Monte Carlo experiments}\label{ss:poutre-exp}

We first compute the survival probability of the beam~\eqref{eq:survival-example} as a function of the number of cycles when the sequence $(P_i)_{i \geq 1}$ is deterministic and constant. For several values of $P$, the resulting survival curves are plotted on Figure~\ref{fig:MC-mean}. As expected, the heavier the vehicle, the faster the failure.

\begin{figure}[htpb]
    \centering
    \includegraphics[width=.8\textwidth]{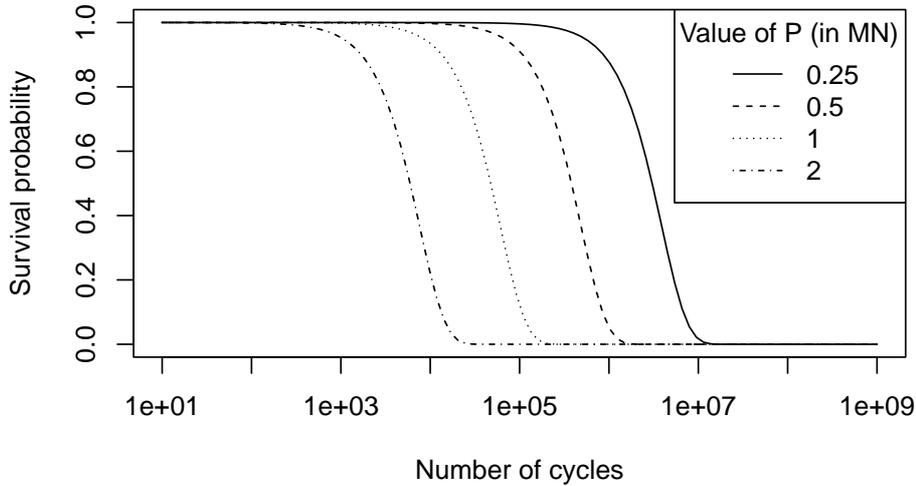}
    \caption{Survival probability of the I-beam, for sequences $(P_i)_{i \geq 1}$ which take the constant value $P$.}
    \label{fig:MC-mean}
\end{figure}

We then investigate the effect of a possible randomness of $(P_i)_{i \geq 1}$, and more precisely the effect of an increasing variance while keeping the expectation fixed. To this aim, we assume that the elements of this sequence are independent realisations of a random variable $P$ such that $P^\alpha$ is Gamma distributed, with rate $\theta>0$ and shape parameter $a>0$. This choice has two motivations: unlike a Gaussian model, it ensures positivity, and from a computational point of view, it allows to use a single random draw for sums of the form $\dis \sum_{i=n_1+1}^{n_2} P_i^\alpha$, which remain Gamma distributed with rate $\theta$ and shape parameter $(n_2-n_1)a$. 

Under this assumption, we have
\begin{equation*}
    \Exp^*[P] = \theta^{-1/\alpha} \, \frac{\Gamma(a+1/\alpha)}{\Gamma(a)}, \qquad \frac{\sqrt{\Var^*(P)}}{\Exp^*[P]} = \sqrt{\frac{\Gamma(a+2/\alpha) \, \Gamma(a)}{\Gamma(a+1/\alpha)^2}-1}.
\end{equation*}
We consider several choices of the parameters $(\theta,a)$ which ensure that $\Exp^*[P]$ remains equal to $P_\mathrm{mean} = 0.25~\mathrm{MN}$, while the standard deviation of $P$ is $c \, P_\mathrm{mean}$ for $c \in \{0, 0.2, 0.5, 1\}$. The corresponding densities for $P$ are plotted on Figure~\ref{fig:MC-var}~(a). We plot on Figure~\ref{fig:MC-var}~(b) the survival probability of the structure for these various load parameters. It is observed that a larger variance significantly reduces the life time of the structure. This illustrates the important influence of tail events, namely large values of $P$, on the survival probability.

\begin{figure}[htpb]
    \centering
    \includegraphics[width=.8\textwidth]{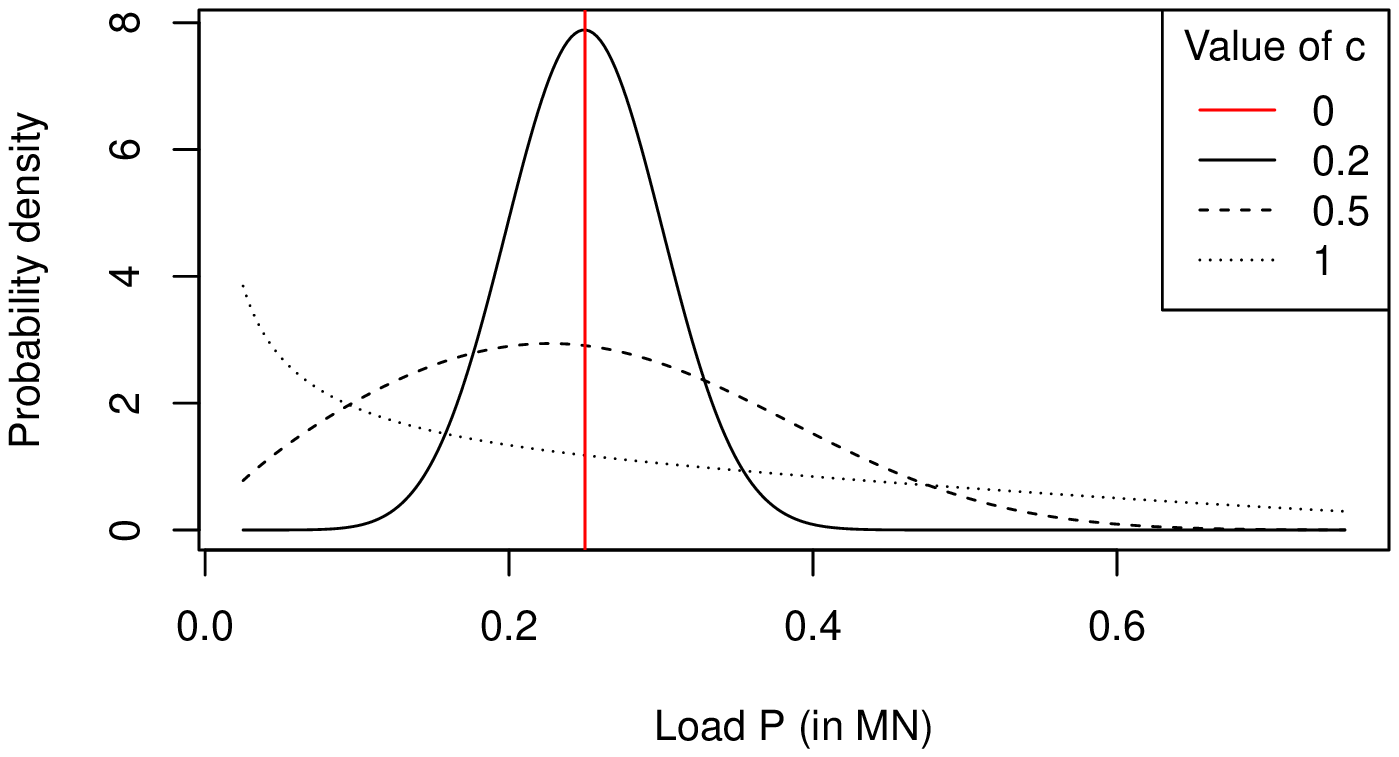}
    \includegraphics[width=.8\textwidth]{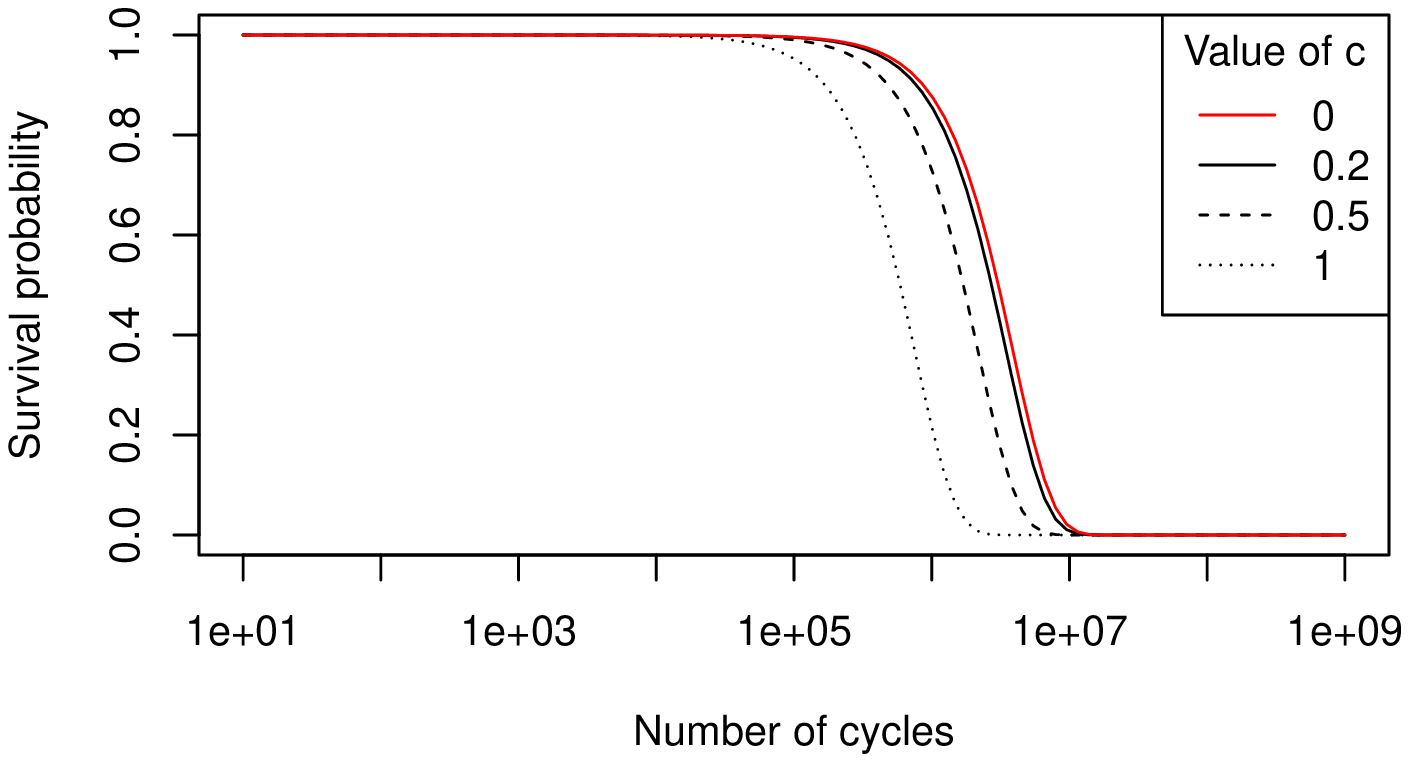}
    \caption{(a): probability density of the random loads $P$, which are adjusted such that they all share the same expectation $P_\mathrm{mean}$ and such that their variation coefficient is $c$. (b) survival probability of the I-beam, for sequences $(P_i)_{i \geq 1}$ distributed according to the density shown in (a) for a given value of $c$.}
    \label{fig:MC-var}
\end{figure}

The resemblance between the survival curves of Figure~\ref{fig:MC-mean} (deterministic and constant load with varying value) and Figure~\ref{fig:MC-var}~(b) (random load with fixed expectation and varying variance) suggests to define, for a random load with expectation $P_\mathrm{mean}$ and coefficient of variation $c$, a notion of \emph{deterministic equivalent load} $P_\mathrm{eq}(c)$, which is (if it exists) the deterministic and constant load for which the associated survival curve fits the survival curve of the random load. Since these curves need not overlap exactly, we define the deterministic equivalent load $P_{\mathrm{eq},p}(c)$ as a function of the reference probability $p \in (0,1)$ as the deterministic and constant load for which the quantile of order $p$ of the structure's NCF coincides with the quantile of order $p$ of the structure's NCF under random load. 

In other words, using the formula~\eqref{eq:survival-example}, for any $p \in (0,1)$:
\begin{itemize}
    \item we define the quantile of order $p$ of the NCF under deterministic and constant load $P$ by
    \begin{equation*}
        N_{\mathrm{det},p}(P) = \inf \left\{n \geq 0: \, \exp\big(-\mQ \left(n P^\alpha\right)^m\big) \leq 1-p\right\};
    \end{equation*}
    \item we define the quantile of order $p$ of the NCF under random load with mean $P_\mathrm{mean}$ and coefficient of variation $c$ by
    \begin{equation*}
        N_{\mathrm{sto},p}(P_\mathrm{mean},c) = \inf\left\{n \geq 0: \, \Exp^*\left[\exp\left(-\mQ \left(\sum_{i=1}^n P_i^\alpha\right)^m\right)\right] \leq 1-p\right\}.
    \end{equation*}
\end{itemize} 
The deterministic equivalent load $P_{\mathrm{eq},p}(c)$ is then defined by the equation
\begin{equation} \label{eq:def_P_equiv}
    N_{\mathrm{det},p}(P_{\mathrm{eq},p}(c)) = N_{\mathrm{sto},p}(P_\mathrm{mean},c).
\end{equation}
For a fixed value of $P_\mathrm{mean} = 0.25~\mathrm{MN}$, the value of the ratio $P_{\mathrm{eq},p}(c)/P_\mathrm{mean}$, as a function of the coefficient of variation $c$, and for the two reference probabilities $p=0.05$ and $p=0.5$, is plotted on Figure~\ref{fig:MC-eq}. For the values of interest of $m$, $\mQ$ and $P_\mathrm{mean}$, we observe that the curves corresponding to the two reference probabilities $p=0.05$ and $p=0.5$ are almost identical, but this need not be a general fact. We note that $P_{\mathrm{eq},p}(c)/P_\mathrm{mean}$ is always larger than 1 and increases when $c$ increases, which underlines again the fact that large load values (which may occur in the random context) have a critical effect on the survival probability.

\begin{figure}[htpb]
    \centering
    \includegraphics[width=.8\textwidth]{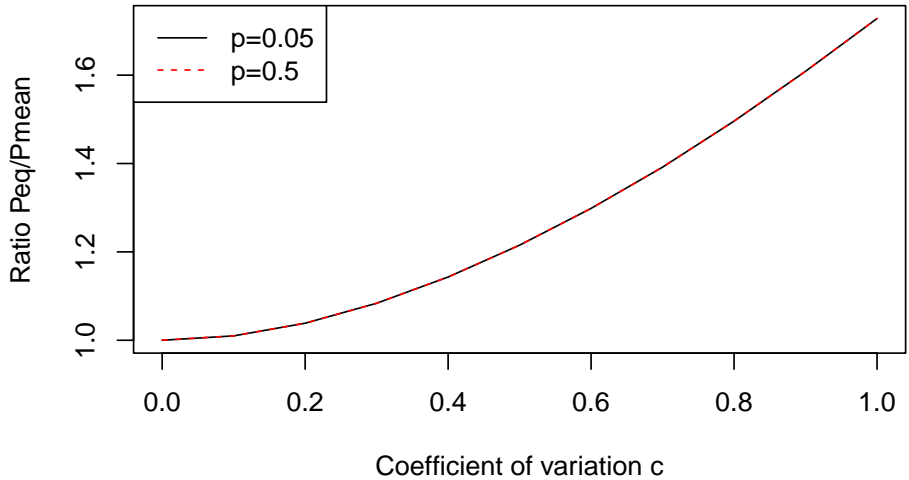}
    \caption{Value of the ratio $P_{\mathrm{eq},p}(c)/P_\mathrm{mean}$ as a function of $c$, where $P_{\mathrm{eq},p}(c)$, defined by~\eqref{eq:def_P_equiv}, is the deterministic equivalent load (with respect to the reference probability $p$) corresponding to a random load with mean $P_\mathrm{mean}$ and coefficient of variation $c$.} 
    \label{fig:MC-eq}
\end{figure}

\subsection{Density of the failure point}\label{ss:poutre-fail}

We conclude this section by investigating the density of failure point, which is given by~\eqref{eq:failure_density}. It may be checked that $\Sevu(x,y,z)$ is invariant with respect to any of the transforms $x \mapsto L-x$, $y \mapsto -y$ and $z \mapsto -z$. Besides, for any $(x,y)$, the function $z \mapsto \Sevu(x,y,z)$ reaches its maximum at $z=0$. The points with maximal failure density, as introduced in Section~\ref{sss:fail}, are therefore located in the $z=0$ plane and symmetrically distributed with respect to the centre of the beam. On Figure~\ref{fig:p-fail}, we represent the failure density $p_\mathrm{fail}(x,y,0)$ in the quarter plane $0 \leq x \leq L/2$ and $0 \leq y \leq h/2$.

\begin{figure}[htpb]
    \centering
    \includegraphics[width=\textwidth]{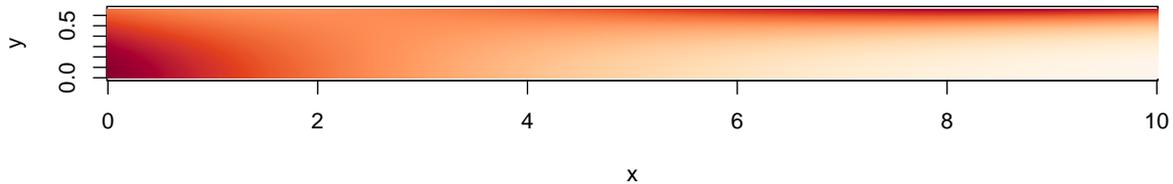}
    \caption{The failure density $p_\mathrm{fail}(x,y,z)$ for $0 \leq x \leq L/2$, $0 \leq y \leq h/2$ and $z=0$. The darker the colour, the higher the value of the density. Here, the most probable points of failure are on the left-hand side of the beam ($x \approx 0$) and on the right top ($8 \leq x \leq 10$, $y \approx h/2$).}
    \label{fig:p-fail}
\end{figure}

\appendix

\section{Laplace approximation of the integral term in the Weibull--Basquin model}\label{app:laplace}

The formula~\eqref{eq:indy3} for the survival probability of a structure in the Weibull--Basquin model derived in Section~\ref{ss:WB-struct} requires to evaluate the integral
\begin{equation*}
    \mI = \int_{x \in E} \Sevu(x)^{\alpha m} \, \lambda(\dd x)
\end{equation*}
in the computation of the constant $\mQ$ defined by~\eqref{eq:indy4}. This may be computationally costly if a finite element method with a small mesh size is employed. On the other hand, as soon as $\alpha m$ is large, one may expect only the largest values of $\Sevu(x)^{\alpha m}$ to actually contribute to this integral. Thus, the mere knowledge of $\Sevu(x)$ in the neighbourhood of its local maxima --- which we shall call \emph{hot points} ---, and no longer in the whole structure $E$, might be sufficient to provide a satisfactory approximation of $\mI$. This idea lies at the basis of the \emph{Laplace method}, which provides an asymptotic approximation of the form
\begin{equation*}
    \mI \approx \sum_{j \geq 1} \Sevu(x^*_j)^{\alpha m} \, V^*_j(\alpha m), 
\end{equation*}
in the regime $\alpha m \gg 1$. In this expression, the quantities $\{V^*_j(\alpha m), \, j \geq 1\}$ are expressed in volume units and are obtained by performing a Taylor expansion, up to the first nonvanishing derivatives, of 
\begin{equation*}
    \varphi(x) = \ln \Sevu(x)
\end{equation*}
in the neighbourhood of the hot points $\{x^*_j, \, j \geq 1\}$.

We illustrate this approach on the example of the I-beam presented in Section~\ref{s:poutre}. A first remark in this context is that the unitary severity field\footnote{We now denote points of $E$ by the triple $(x,y,z)$ rather than by the symbol $x$.} $\Sevu(x,y,z)$ is invariant with respect to each of the transforms $x \mapsto L-x$, $y \mapsto -y$ and $z \mapsto -z$, so that the computation may be restricted to the set
\begin{equation*}
    E' = \left\{(x,y,z): \, 0 \leq x \leq L/2, \, 0 \leq y \leq h/2, \, 0 \leq z \leq I(y) \right\}, \quad I(y) = \begin{cases}
        f/2 & \text{if $y \leq h/2-e$,}\\
        b/2 & \text{if $h/2-e < y \leq h/2$,}
    \end{cases}
\end{equation*}
and we have
\begin{equation*}
    \mI = 8 \, \mI', \qquad \mI' = \int_{(x,y,z) \in E'} \exp\left(k \, \varphi(x,y,z)\right) \, \dd x \, \dd y \, \dd z,
\end{equation*}
where we have set $k = \alpha m$ and $\varphi(x,y,z) = \ln \Sevu(x,y,z)$.

It may be checked that the unitary severity field $\Sevu(x,y,z)$ admits two local maxima in $E'$, namely
\begin{equation*}
    (x^*_1,y^*_1,z^*_1) = (0,0,0), \qquad (x^*_2,y^*_2,z^*_2) = (x^*_2,h/2,0),
\end{equation*}
for some $0 < x^*_2 < L/2$. These two hot points are represented in the $z=0$ plane on Figure~\ref{fig:laplace}.

At the first hot point $(x^*_1,y^*_1,z^*_1) = (0,0,0)$, one has
\begin{align*}
    &\partial_x \varphi(x^*_1,y^*_1,z^*_1) < 0,\\
    &\partial_y \varphi(x^*_1,y^*_1,z^*_1) = 0, \quad \partial_{yy} \varphi(x^*_1,y^*_1,z^*_1) < 0,\\
    &\partial_z \varphi(x^*_1,y^*_1,z^*_1) < 0.
\end{align*}
We therefore define
\begin{equation*}
    \mI'_1 = \int_{(x,y,z) \in E'} \ee^{k(\varphi(x^*_1,y^*_1,z^*_1) + \partial_x \varphi(x^*_1,y^*_1,z^*_1)(x-x^*_1) + \frac{1}{2} \partial_{yy} \varphi(x^*_1,y^*_1,z^*_1)(y-y^*_1)^2 + \partial_z \varphi(x^*_1,y^*_1,z^*_1)(z-z^*_1))} \, \dd x\dd y\dd z.
\end{equation*}
Introducing the notation
\begin{equation*}
    \Delta x^*_1 = \frac{1}{-\partial_x \varphi(x^*_1,y^*_1,z^*_1)}, \qquad \Delta y^*_1 = \frac{1}{\sqrt{-\partial_{yy} \varphi(x^*_1,y^*_1,z^*_1)}}, \qquad \Delta z^*_1 = \frac{1}{-\partial_z \varphi(x^*_1,y^*_1,z^*_1)},
\end{equation*}
one may compute
\begin{equation*}
    \mI'_1 = \Sevu(x^*_1,y^*_1,z^*_1)^k \, \left(V^*_{1,\mathrm{web}}(k) + V^*_{1,\mathrm{flange}}(k)\right),
\end{equation*}
with
\begin{align*}
    V^*_{1,\mathrm{web}}(k) &= \frac{\Delta x^*_1}{k} \,  \Phi_{(1)}\left(\frac{L/2}{\Delta x^*_1/k}\right) \times \frac{\Delta y^*_1}{\sqrt{k}} \,  \Phi_{(2)}\left(\frac{h/2-e}{\Delta y^*_1/\sqrt{k}}\right) \times \frac{\Delta z^*_1}{k} \, \Phi_{(1)}\left(\frac{f/2}{\Delta z^*_1/k}\right),
    \\
    V^*_{1,\mathrm{flange}}(k) &= \frac{\Delta x^*_1}{k} \, \Phi_{(1)}\left(\frac{L/2}{\Delta x^*_1/k}\right) \\
    &\qquad \qquad \times \frac{\Delta y^*_1}{\sqrt{k}} \, \left[\Phi_{(2)}\left(\frac{h/2}{\Delta y^*_1/\sqrt{k}}\right) - \Phi_{(2)}\left(\frac{h/2-e}{\Delta y^*_1/\sqrt{k}}\right) \right] \times \frac{\Delta z^*_1}{k} \, \Phi_{(1)}\left(\frac{b/2}{\Delta z^*_1/k}\right),
\end{align*}
where the functions $\Phi_{(1)}$ and $\Phi_{(2)}$ are defined by
\begin{equation*}
    \Phi_{(1)}(u) = \int_{v=0}^u \ee^{-v} \, \dd v, \qquad \Phi_{(2)}(u) = \int_{v=0}^u \ee^{-v^2/2} \, \dd v.
\end{equation*}

For the second hot point $(x^*_2,y^*_2,z^*_2) = (x^*_2,h/2,0)$, we have
\begin{align*}
    &\partial_x \varphi(x^*_2,y^*_2,z^*_2) = 0, \quad \partial_{xx} \varphi(x^*_2,y^*_2,z^*_2) < 0,\\
    &\partial_y \varphi(x^*_2,y^*_2,z^*_2) > 0,\\
    &\partial_z \varphi(x^*_2,y^*_2,z^*_2) < 0,
\end{align*}
which leads us to define
\begin{equation*}
    \mI'_2 = \int_{(x,y,z) \in E'} \ee^{k(\varphi(x^*_2,y^*_2,z^*_2) + \frac{1}{2}\partial_{xx} \varphi(x^*_2,y^*_2,z^*_2)(x-x^*_2)^2 - \partial_y \varphi(x^*_2,y^*_2,z^*_2)(y^*_2-y) + \partial_z \varphi(x^*_2,y^*_2,z^*_2)(z-z^*_2))} \, \dd x\dd y\dd z.
\end{equation*}
We observe that
\begin{equation*}
    \mI'_2 = \Sevu(x^*_2,y^*_2,z^*_2)^k \, \left(V^*_{2,\mathrm{web}}(k) + V^*_{2,\mathrm{flange}}(k)\right)
\end{equation*}
with
\begin{align*}
    V^*_{2,\mathrm{web}}(k) &= \frac{\Delta x^*_2}{\sqrt{k}} \, \left[ \Phi_{(2)}\left(\frac{x^*_2}{\Delta x^*_2/\sqrt{k}}\right) + \Phi_{(2)}\left(\frac{L/2-x^*_2}{\Delta x^*_2/\sqrt{k}}\right) \right]
    \\
    &\qquad \qquad \times \frac{\Delta y^*_2}{k} \, \left[ \Phi_{(1)}\left(\frac{h/2}{\Delta y^*_2/k}\right) - \Phi_{(1)}\left(\frac{e}{\Delta y^*_2/k}\right) \right] \times \frac{\Delta z^*_2}{k} \, \Phi_{(1)}\left(\frac{f/2}{\Delta z^*_2/k}\right),
    \\
    V^*_{2,\mathrm{flange}}(k) &= \frac{\Delta x^*_2}{\sqrt{k}} \, \left[ \Phi_{(2)}\left(\frac{x^*_2}{\Delta x^*_2/\sqrt{k}}\right) + \Phi_{(2)}\left(\frac{L/2-x^*_2}{\Delta x^*_2/\sqrt{k}}\right) \right]
    \\
    & \qquad \qquad \times \frac{\Delta y^*_2}{k} \, \Phi_{(1)}\left(\frac{e}{\Delta y^*_2/k}\right) \times \frac{\Delta z^*_2}{k} \, \Phi_{(1)}\left(\frac{b/2}{\Delta z^*_2/k}\right),
\end{align*}
and
\begin{equation*}
    \Delta x^*_2 = \frac{1}{\sqrt{-\partial_{xx} \varphi(x^*_2,y^*_2,z^*_2)}}, \qquad \Delta y^*_2 = \frac{1}{\partial_y \varphi(x^*_2,y^*_2,z^*_2)}, \qquad \Delta z^*_2 = \frac{1}{-\partial_z \varphi(x^*_2,y^*_2,z^*_2)}.
\end{equation*}

\medskip

The expressions defining $V^*_{1,\mathrm{web}}$ and $V^*_{1,\mathrm{flange}}$ (respectively $V^*_{2,\mathrm{web}}$ and $V^*_{2,\mathrm{flange}}$) involve the \emph{characteristic lengths} $\Delta x^*_1/k$, $\Delta y^*_1/\sqrt{k}$ and $\Delta z^*_1/k$ (respectively $\Delta x^*_2/\sqrt{k}$, $\Delta y^*_2/k$ and $\Delta z^*_2/k$), which indicate at which scale the value of the unitary severity can be considered to be close to its maximum $\Sevu(x^*_1,y^*_1,z^*_1)$ (respectively $\Sevu(x^*_2,y^*_2,z^*_2)$). These characteristic lengths are represented on Figure~\ref{fig:laplace} (in the $x$ and $y$ directions) for various values of $k$.

\begin{figure}[htpb]
    \includegraphics[width=\textwidth]{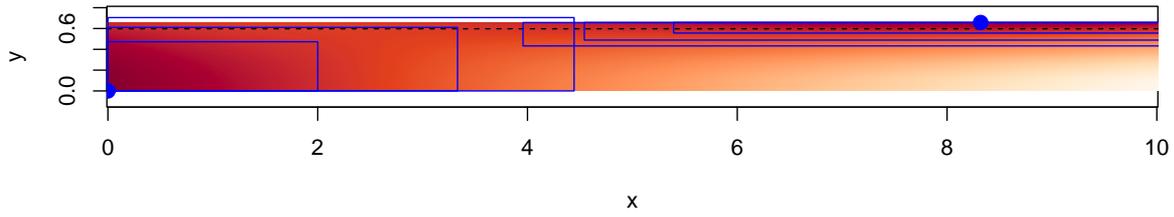}
    \caption{The background picture represents the unitary severity field $\Sevu(x,y,0)$ for $x \in [0,L/2]$ and $y \in [0,h/2]$. The darker the colour, the larger the severity. A dashed line is added at $y=h/2-e$ to show the separation between the web and the flange. The two hot points $(x^*_1=0,y^*_1=0,z^*_1=0)$ and $(x^*_2,y^*_2=h/2,z^*_2=0)$ are represented with blue points. For $k \in \{4.5, 6, 10\}$, the rectangles $[x^*_1, x^*_1 + \Delta x^*_1/k] \times [y^*_1, y^*_1 + \Delta y^*_1/\sqrt{k}]$ and $[x^*_2 - \Delta x^*_2/\sqrt{k}, x^*_2 + \Delta x^*_2/\sqrt{k}] \times [y^*_2-\Delta y^*_2/k, y^*_2]$ are added in blue. 
    }
    \label{fig:laplace}
\end{figure}

The Laplace approach eventually consists in approximating $\mI'$ by the sum $\mI'_1+\mI'_2$. For various values of $k$, the ratio between the approximation $\mI'_1+\mI'_2$ and the reference value $\mI'$\footnote{The reference value for $\mI'$ is computed using a quadrature rule on a fine spatial discretisation of the unitary severity field (in practice, we consider a grid with mesh size $2~\mathrm{cm}$ in $x$, $5~\mathrm{mm}$ in $y$, and $2~\mathrm{mm}$ in $z$ in the web, $1~\mathrm{cm}$ in $z$ in the flange).} is given in bold in the first column of Table~\ref{tab:laplace}. The next columns respectively display the ratio $\mI'_1/(\mI'_1+\mI'_2)$, $V^*_{1,\mathrm{web}}/(V^*_{1,\mathrm{web}}+V^*_{1,\mathrm{flange}})$ and $V^*_{2,\mathrm{web}}/(V^*_{2,\mathrm{web}}+V^*_{2,\mathrm{flange}})$, in order to quantify the dominant contributions in $\mI'_1+\mI'_2$.

\begin{table}[htpb]
    \centering
    \begin{tabular}{ccccc}
        $k$ & $\displaystyle \frac{\mI'_1+\mI'_2}{\mI'}$ & $\displaystyle \frac{\mI'_1}{\mI'_1+\mI'_2}$ & $\displaystyle \frac{V^*_{1,\mathrm{web}}}{V^*_{1,\mathrm{web}}+V^*_{1,\mathrm{flange}}}$ & $\displaystyle \frac{V^*_{2,\mathrm{web}}}{V^*_{2,\mathrm{web}}+V^*_{2,\mathrm{flange}}}$\\[4mm]
        \hline
        $4.5$ & $\mathbf{1.21}$ & $0.36$ & $0.20$ & $ 0.05$\\
        $6$ & $\mathbf{1.11}$ & $0.32$ & $0.22$ & $0.04$\\
        $10$ & $\mathbf{0.96}$ & $0.24$ & $0.27$ & $0.02$
    \end{tabular}
    \caption{Numerical results for the Laplace approximation of $\mI'$ by $\mI'_1+\mI'_2$ in the example of the I-beam from Section~\ref{s:poutre}.}
    \label{tab:laplace}
\end{table}

As a conclusion, for the value $k = \alpha m = 4.5$ of the example of Section~\ref{s:poutre}, the integral $\mI$ can be estimated by the Laplace method with a relative accuracy of $21\%$. This method only requires to compute the unitary severity field and its first nonvanishing derivatives at the hot points, which may represent a substantial computational gain with respect to a complete discretisation method, which would require to compute the unitary severity field in the whole structure. Furthermore, and as expected, the accuracy of the method improves (i.e. the relative error decreases) when $\alpha m$ assumes larger values (the error being already of only $11\%$ for $\alpha m = 6$).

\section*{Acknowledgements}

The work presented in this article elaborates on a preliminary work that explored some of the issues considered here, and which was performed in the context of the internship of Bruno Martins Aboud at École des Ponts ParisTech. We would also like to thank V\'eronique Le~Corvec and Jorge Semiao at OSMOS Group for stimulating and enlightening discussions about the work reported in this article. This research received the support of OSMOS Group, as part of its effort to develop new solutions for the Structural Health Monitoring of civil and industrial assets.

\bibliographystyle{plain}
\bibliography{biblio-miner}

\begin{thebibliography}{10}

\bibitem{Bas10}
O.~H. Basquin.
\newblock The exponential law of endurance tests.
\newblock In {\em Proc. Am. Soc. Test. Mater.}, volume~10, pages 625--630,
  1910.

\bibitem{Bog07}
V.~I. Bogachev.
\newblock {\em Measure theory. {V}ol. {I}, {II}}.
\newblock Springer-Verlag, Berlin, 2007.

\bibitem{BomMaySch97}
H.~Bomas, P.~Mayr, and M.~Schleicher.
\newblock Calculation method for the fatigue limit of parts of case hardened
  steels.
\newblock {\em Materials Science and Engineering: A}, 234:393--396, 1997.

\bibitem{CasFer09}
E.~Castillo and A.~Fern{\'a}ndez-Canteli.
\newblock {\em A unified statistical methodology for modeling fatigue damage}.
\newblock Springer Science \& Business Media, 2009.

\bibitem{Cas06}
E.~Castillo, M.~López-Aenlle, A.~Ramos, A.~Fernández-Canteli, R.~Kieselbach,
  and V.~Esslinger.
\newblock Specimen length effect on parameter estimation in modelling fatigue
  strength by {W}eibull distribution.
\newblock {\em International Journal of Fatigue}, 28:1047--1058, 2006.

\bibitem{Bar19}
J.~Fonseca~Barbosa, J.~A.F.O. Correia, R.C.S. Freire~Junior, S.-P. Zhu, and
  A.~M.P. De~Jesus.
\newblock Probabilistic {SN} fields based on statistical distributions applied
  to metallic and composite materials: State of the art.
\newblock {\em Advances in Mechanical Engineering}, 11(8):1687814019870395,
  2019.

\bibitem{Fou14}
R.~Fouchereau, G.~Celeux, and P.~Pamphile.
\newblock Probabilistic modeling of {S--N} curves.
\newblock {\em International Journal of Fatigue}, 68:217--223, 2014.

\bibitem{Jeu21}
D.~Jeulin.
\newblock {\em Morphological models of random structures}.
\newblock Springer, 2021.

\bibitem{Min45}
M.~A. Miner.
\newblock Cumulative damage in fatigue.
\newblock {\em J. Appl. Mech.}, 12(3):A159--A164, 1945.

\bibitem{Mir20}
E.~Miranda.
\newblock {\em Mod{\'e}lisation et caract{\'e}risation des risques extr{\^e}mes
  en fatigue des mat{\'e}riaux}.
\newblock PhD thesis, Sorbonne Universit{\'e}, 2020.

\bibitem{Pal24}
A.~G. Palmgren.
\newblock Die {L}ebensdauer von {K}ugellargern.
\newblock {\em Zeitschrift des Vereines Deutscher Ingenieure}, 68(4):339--341,
  1924.

\bibitem{Thi08}
E.~Thieulot-Laure.
\newblock {\em M{\'e}thode probabiliste unifi{\'e}e pour la pr{\'e}diction du
  risque de rupture en fatigue}.
\newblock PhD thesis, {\'E}cole Normale Sup{\'e}rieure de Cachan -- ENS Cachan,
  2008.

\bibitem{Wei51}
W.~Weibull.
\newblock A statistical distribution function of wide applicability.
\newblock {\em J. Appl. Mech.}, 18(3):293--297, 1951.

\bibitem{Woh70}
A.~W{\"o}hler.
\newblock {\em {\"U}ber die Festigkeitsversuche mit Eisen und Stahl}.
\newblock Ernst \& Korn, 1870.

\end{thebibliography}

\end{document}